\documentclass[10pt]{article}
\usepackage[english]{babel}
\usepackage[fleqn]{amsmath}
\usepackage{amssymb}
\usepackage{amsthm}
\usepackage{psfrag}
\usepackage{graphicx}
\usepackage{fancybox}
\usepackage[margin=10pt,font=footnotesize,labelfont=bf]{caption}
\usepackage[utf8]{inputenc}
\usepackage[T1]{fontenc}
\usepackage{lmodern}
\usepackage{hyperref}

\title{Sheaf Logic, Quantum Set Theory and the Interpretation of Quantum Mechanics.}

\author{J. Benavides\footnote{Department of Mathematics, Ulisse Dini, University of Florence, Italy. navarro@math.unifi.it}}

\newtheorem{defin}{Definition}[section]
\newtheorem{teor}{Theorem}[section]
\newtheorem{corol}{Corollary}[section]
\newtheorem{lema}{Lemma}[section]

\begin{document}

\maketitle

\begin{abstract}
Based on the Sheaf Logic approach to set theoretic forcing,  a hierarchy of Quantum Variable Sets is constructed which generalizes and simplifies the analogous construction developed  by Takeuti on boolean valued models of set theory. Over this model two alternative proofs of Takeuti's correspondence, between self adjoint operators and the real numbers of the model, are given. This approach results to be more constructive showing a direct relation with the Gelfand representation theorem, revealing also  the importance of these results with respect to the interpretation of Quantum Mechanics in close connection with the Deutsch-Everett multiversal interpretation. Finally, it is shown how in this context the notion of genericity and the corresponding generic model theorem can help to explain the emergence of classicality  also in connection with the Deutsch-Everett perspective. 
\end{abstract}

\section{Introduction}

The advent of Quantum Mechanics (QM) and the problems linked to its ontological nature changed our conception of physical reality in a radical way, problems that before had concerned just philosophers of science became central to the physics debate.  Notions like observable, observer  and measurement, which  had not been problematic to the conception of physical theories, became fundamental and subject of numerous controversies. Furthermore, the classical realist conception of physical objects as independent bearers of properties, on which the ontology of classical theories was based, was also challenged;  becoming evident the interpretational  difficulties of the theory, a problem  which physics had never affronted before.\\
\indent Today, more than one hundred years after Max Planck formulated the quantum hypothesis,  we still do not have a settled agreement  about what quantum reality is or if there is something as a quantum reality at all. Nevertheless, the theory has been incredibly successful in its predictive role. For this reason, many physicists  think that probably there is no necessity of an interpretation of the theory that go further than the interpretation linked to its predictive nature. But, despite this predictive success, it has not been possible to conciliate  the theory in its instrumentalist form with  General Relativity (GR), and it is becoming evident that a new formalism, which give us a completely new perspective of the theory, will be needed to solve this problem.\\
\indent Numerous proposals have been advanced to solve the interpretational issues of QM, however,  none of these  have been able to transcend  their heuristic argumentations with a solid mathematical machinery that captures and go further than the classical tools, settling the respective interpretation. Our inability to do so  is probably a sign that behind the understanding of QM hides the necessity to transcend also the classical mathematical formalism that lies at the foundations of the tools used so far to conceive physical theories.\\

In recent years Topos Theory has captured the attention of people working on the foundations of QM as a possible route to reformulate the theory in a way flexible enough to include  relativistic concepts, and where a definite interpretation could be finally settled. The origins of this approach can be traced back to the work developed by Takeuti in 1975 on Boolean Valued Models of Set Theory. Takeuti  proved that in a Boolean Valued model of set theory constructed over a complete boolean algebra of projector operators of a Hilbert space, there is a correspondence between the self adjoint operators which spectral family is contained in the boolean algebra of projections and the real numbers of the boolean valued model \cite{takeuti}\cite{ozawa}.  In those days the importance of this result  respect to its possible relation with interpretational issues of QM was discussed (see \cite{davis})  but no conclusive results were obtained, and maybe due to the fact that the result used  advanced tools of set theory and logic, it did not capture the attention of the physics community.\\
\indent  Recently the work of  C. Isham, A. Doering  \cite{isham} \cite{doering} and others have  brought the attention back to these methods and particularly to the idea that these tools can be used to obtain a new conception of the continuum useful to formulate Quantum Gravity theories and to obtain a new perspective of QM.  Even if not explicitly stated both, Takeuti's and Isham-Doering's approaches, are a reformulation of the old idea proposed several times after the publication of the seminal paper in Quantum Logic (QL) by Von Neumann and Birkhoff, regarding the necessity of a formalism founded over a Quantum Logic as a route to reformulate the theory in a way able to capture the essence of quantum reality. However, it is still not clear that this new approach will give us a better understanding of the theory. Unfortunately, the intrinsic difficulties of Cohen's forcing in the boolean formulation and the abstract categorical machinery of topos theory have obscured the potentiality of these methods to obtain a better picture of QM.\\

In 1995 X. Caicedo introduced what can be considered so far the most user-friendly approach to  Kripke-Joyal Semantics \cite{caicedo} (the semantics intrinsic to a topos), giving a new perspective that avoids the technicalities linked to the categorical tools of topoi. Caicedo introduced a model theory  of variable structures where it is possible to introduce a definition of genericity and a generic model theorem which unifies set theoretic forcing constructions and the classical theorems of model theory as completeness, compactness, omitting types etc..  In this context the approach to set theoretic forcing  generalizes the Scott-Solovay Boolean \cite{scott} and the Fitting intuitionistic  \cite{fitting} formulations, simplifying the constructions and clarifying the  essence of the proofs of  classical independence results as the independence of the continuum hypothesis \cite{benavides} \cite{benavides3}. Another remarkable fact  is how  interesting connections between classical logic and intuitionistic logic are revealed, showing that the logic of sheaves is not just intuitionistic but that constitutes a continuum of logics between classical and intuitionistic logic, where the independence results of set theory and the classical theorems of model theory can be conceived as  a consequence of some limit process over this continuum.\\ 
\indent In this work I apply these tools to QM, the idea  in a few words is to show that the logic intrinsic to QM lies in this continuum of logics and then to show that the emergence of classicality can be conceived as a limit process over these logics. The logic used here will differ from the classical QL of Von Neumann and Birkhoff; its construction arises from the local character of truth of  Sheaf Logic, which allows to introduce some contextual features of QM as those derived  from the Kochen-Specker theorem   and  from the Deutsch-Everett multiversal interpretation of interference phenomena.  Over this logic a hierarchy of \textit{Quantum Variable Sets} is constructed that generalizes and simplifies the Boolean approach of Takeuti. In this model, two alternative proofs of Takeuti's correspondence, between self adjoint operators and the real numbers of the model, are given. This approach results to be more constructive, showing a direct relation with the Gelfand representation theorem,  and revealing also  the importance of these results with respect to the interpretation of QM in close connection with the Deutsch-Everett multiversal interpretation of quantum theory. Finally it is shown how the collapse via generic models of this structure of quantum variable sets can help to explain the emergence of classicality  also in close relation with the Deutsch-Everett perspective\footnote{I have been recently referred to the work of W. Boos \cite{boos} and R. A. Van Wesep \cite{vanw} which suggest the use of the  generic property in an analogous way as proposed here  to explain the emergence of classicality in QM. Even if Boos work is motivated in analogous ideas it is in essence different because it is based on the use of measure algebras and not in boolean algebras of projections. On the other hand Van Wesep approach is more closely related to the ideas here presented, it is particularly interesting the study of the emergence of probability that he proposes, however he does not mention Takeuti's result which is fundamental to understand in which sense these tools can explain the collapse to a classical world. Both papers are also based on the classical Boolean approach to Cohen's forcing. As I am suggesting here, Caicedo's work simplifies remarkably the boolean approach in a way  closely related  with topos theoretic tools and that results fundamental to understand how these methods can  help to settle an interpretation of QM. Due to I concluded this work before knowing of the existence of Boos and Van Wesep papers I will postpone the discussion of their results in the context of the tools here presented to a future work.}.\\

I have divided this work in two main sections, the first one is an introduction to sheaf logic as developed in Caicedo's work, giving here an approach oriented to physicists. Even if it is not possible to get a complete picture of these tools without some knowledge of model theory and set theory the guiding ideas are very natural and can be followed without being an expert in these fields. The second part deals with the construction of the hierarchy of Quantum Variable Sets and the results cited above. Hopefully, the ideas contained here will be also useful to people  working with applications of topos theory in physics, particularly bringing attention to the fundamental role of logic that sometimes is forgotten giving prevalence to the geometric character of the theory. As I argue in this work, the role of logic in these tools can be fundamental to obtain a new satisfactory picture of QM and probably the connections between geometry and logic, the most remarkable feature of topoi, will be  fundamental in the construction of a future Quantum Gravity theory.

\section{Sheaves of Structures}

The notion of a  Sheaf of structures has its origins on the study of the continuation of analytic functions in the XIX century, but the modern definition  was introduced just few years after the end of the second world war, by Leray, Cartan and Lazard, in the context of Algebraic topology and Algebraic Geometry. However, the idea of a sheaf of structures is naturally contained in the conception of  spacetime which derives from the Galilean relativity principle. The Galilean relativity principle asserts that the dynamical laws are the same when are referred to any frame in uniform motion. This principle forces us to abandon the Aristotelian Picture of a fixed and absolute background space which constitutes a preferential frame where physical objects move. Thus, the idea of considering a point in space as the same point an instant later looses its meaning. Instead,  Galilean dynamics implies that there is not one fixed Euclidean 3-dimensional space where the physical world is contained, but  that to each instant corresponds  a  different 3-dimensional world \cite{penrose},  which is attached with the other worlds in a continuous way respect to the temporal order(see fig. \ref{Gspace}). 

\begin{figure}
\centering
\includegraphics[scale=0.50]{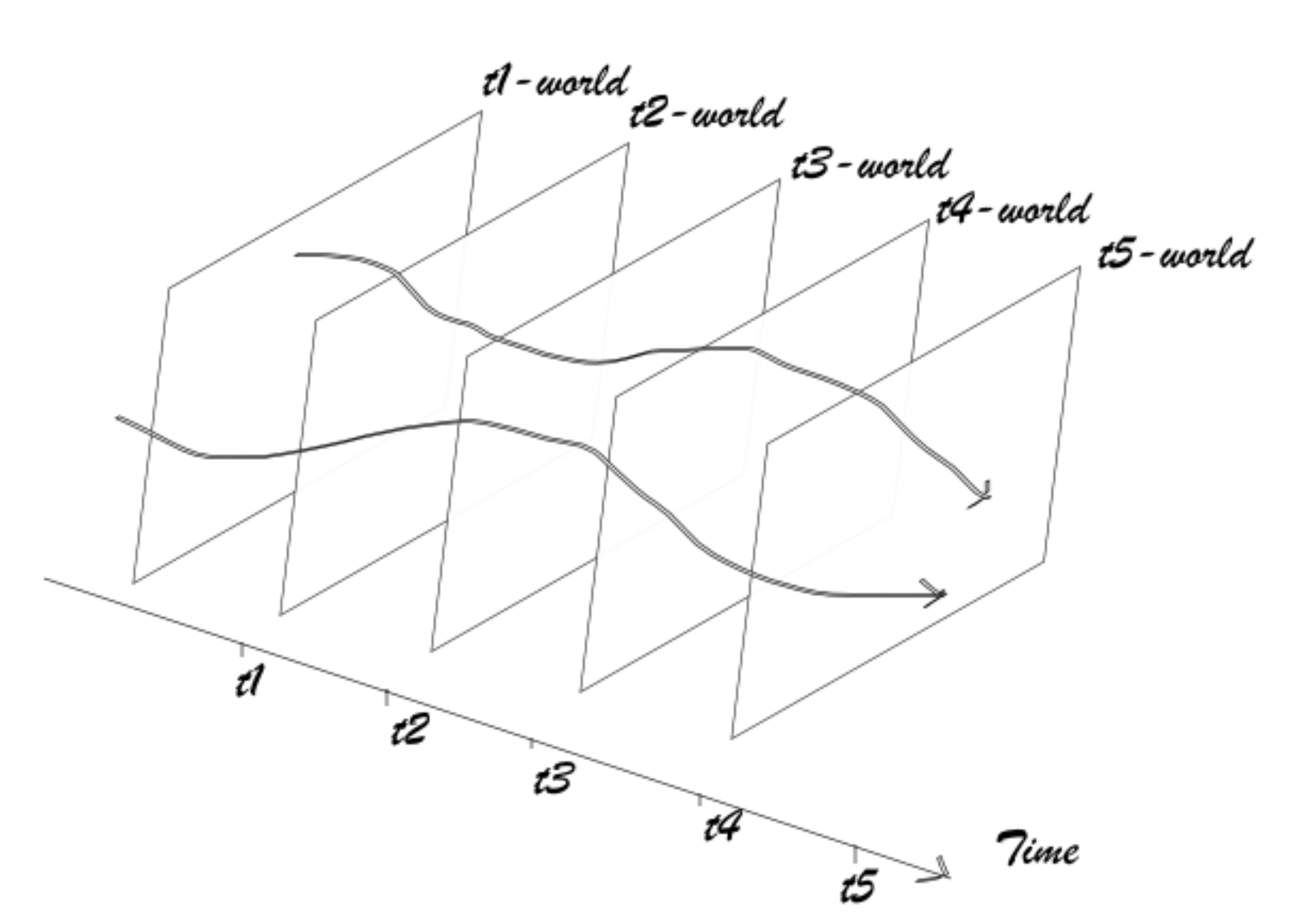}
\caption{Galilean Spacetime}
\label{Gspace}
\end{figure}

Each of these worlds is formed by objects which represent a snapshot of an extended object  in time. There also exist functions and relations defined over each world, which represent and describe the instantaneous physical attributes of these extended objects.  The continuous attachment refers to the possibility of seeing these extended objects as functions that at each time assign an object in the respective world in such a way that they form a continuous object as seeing from the perspective of the hole attachment. In the same way the functions and relations that give instantaneous descriptions attach also in a continuous way when seen as functions and relations of the attachment of the different instantaneous worlds.\\

We can summarize this description of Galilean spacetime as follows. We have a topological space $X=$\textit{temporal line}, for each element $x\in X$ there is an structure $\mathfrak{A}_x$ which is formed by a world $E_x$ which objects constitute a snapshot  of  extended objects in time, functions $f^{x}_1,f^{x}_2,...$ and relations $R^{x}_1,R^{x}_2,...$ which give the instantaneous properties of the extended objects at the instant $x$. Even if at each instant  $x$ we have different worlds $E_x$, this worlds are of a same kind in the sense that the functions, relations and objects can be given  analogous interpretations at each instant. The different worlds $E_x$ attach in an extended universe $E$ in such a way that the attachment of the objects, functions and relations are continuous as seeing as extended objects, functions or relations defined in this extended world. In a few words at each instant $x$ we have an structure $\mathfrak{A}_x=(E_x,R^{x}_1, R^{x}_2...,f^{x}_1,f^{x}_2...)$ formed by a world of  instantaneous snapshots of extended objects, functions and relations that attach in a continuous way. All this features  can be easily formalized and generalized using the notion of a Sheaf of Structures. The ideas presented in this section follow \cite{caicedo},  the proofs of the results presented here can be found there, however  I will try to motivate the results and the main ideas can be understood without knowing all the technical details.

\subsection{Sheaves and Presheaves}

\begin{defin}
Let $X$ be a topological space. A sheaf over $X$ is a couple $(E,p)$, where $E$ is a topological space and $p:E\rightarrow X$ is a local homeomorphism,  or in other words  a continuous function such that for each $e\in E$ there exists a neighbourhood $V$ of $e$ such that:
\begin{enumerate}
\item $p(V)$ is open in $X$
\item $p\upharpoonright_{V}:V\rightarrow p(V)$ is an homeomorphism. 
\end{enumerate}
Given an open set $U$ in $X$, a function $\sigma:U\rightarrow E$ such that  $p\circ \sigma=id_{U}$ is called a local section, if $U=X$, $\sigma$ is called  a global section. The set $p^{-1}(x)\subset E$ for $x\in X$ is called the fibre over $x$.
\end{defin}   

Using the sheaf notion it becomes easy to define a generalization which  captures the essence  contained in the picture of  Galilean spacetime as described above. The formalism of the definition given below contains some technicalities of model theory, but the essence of the definition can be understood even if you are not used to the language of this field.

\begin{defin} Given a fix type of structures $\tau=(R_1,...,f_1,...,c_1,...)$ a sheaf of $\tau$-structures
$\frak{A}$ over a topological space $X$ is given by:\\
a-) A sheaf $( E,p)$ over $X$.\\
b-) For each  $x\in X$,  a $\tau$-structure $\mathfrak{A}_x=(E_x,R^{x}_1, R^{x}_2...,f^{x}_1,...,c^{x}_1,...)$,
where  $E_x=p^{-1}(x)$ (the fiber that could be empty) is  the universe of the $\tau$-structure  $\mathfrak{A}_x$, and the following conditions are satisfied: \\
i. $R^{\mathfrak{A}}=\bigcup_x R_x$ is open in $\bigcup_x
E^{n}_{x}$ seeing as subspace of $E^n$, where $R$ is an n-ary relation symbol.\\
ii.  $f^{\mathfrak{A}}=\bigcup_{x}f_x:\bigcup_x
E_{x}^{m}\rightarrow \bigcup_x E_x$ is a continuous function, where $f$ is an $m$-parameter function symbol.\\
iii. $h:X\rightarrow E$  such that  $h(x)=c_x$, where $c$ is a constant symbol, is continuous.
\end{defin}
 
The type of a structure mentioned in the definition above is a concept used in logic and model theory.  In a few words a type $\tau$  is a language with symbols of relations, functions and constants. A $\tau$-structure $\mathfrak{A}$ is  formed by a set of objects $A$ where the different symbols of relations, functions and constants find an interpretation as functions or relations  over this set and the constants as elements of this set. For example if  $\tau=(\widehat{\times},\widehat{1})$, the structure which has as universe  the rational numbers $\mathbb{Q}$, where $\widehat{ \times}$ is interpreted as the multiplication of rational numbers and $\widehat{1}$ is interpreted as the number $1$ is a $\tau$-structure. This can be rather confusing but in some sense is just telling us that we are attaching structures of the same kind. On the other hand the properties (i), (ii),(iii)  determine that the structures attach in a smooth or continuous way.\\
\indent  A sheaf of structures is a space  extended over the base space $X$ of the sheaf  as Galilean spacetime extends over time. The elements of this space will not be the points of $E$ but the sections of the sheaf conceived as extended objects. The single values of these sections represent just  punctual descriptions of the extended object.\\

Let $U\subset X$ be an open subset of the base space, the set of sections defined over $U$, 
\[\mathfrak{A}(U)=\{\sigma:U\rightarrow E : \sigma\text{ continuous  and } \sigma(x)\in E_x\},\]
can be seen also as the universe of an structure of the same type $\tau$ of the sheaf of structures. This follows from the fact that the continuous attachment guarantees that if, for instance,  we have sections $\sigma_1,...,\sigma_n$ defined over an open set $V$ and some relation $R(\sigma_1(x),...,\sigma_n(x))$ holds at the node $x\in V$, taking $U'=p^{-1}(Im((\sigma_1,...,\sigma_n))\cap R^{\mathfrak{A}})$ which is an open set by the definition above, we have that  for $y\in U=U'\cap V$, $R(\sigma_1(y),...,\sigma_n(y))$ holds. Then we can say that $R(\sigma_1,...,\sigma_n)$ holds in $U$ if $R(\sigma_1(y),...,\sigma_n(y))$ for all $y\in U$. In an analogous way if $f^{\mathfrak{A}}(\sigma_1(x),...,\sigma_n(x))=\mu(x)$ for some section $\mu$ defined over $V$ and a $n$-parameter function symbol $f$, by the definition above we have that $f^{\mathfrak{A}}\circ (\sigma_1,...,\sigma_n)$ is continuous, then there exists an open neighbourhood $U$ of $x$ such that $f^{\mathfrak{A}}\circ (\sigma_1,...,\sigma_n)(U)\subseteq Im(\mu)$, thus $f^{\mathfrak{A}}(\sigma_1(y),...,\sigma_n(y))=\mu(y)$ for all $y\in U$. Using this we can define a function $f^{\mathfrak{A}(U)}$ such that $f^{\mathfrak{A}(U)}(\sigma_1,...,\sigma_n)=\mu$ if $f^{\mathfrak{A}}(\sigma_1(y),...,\sigma_n(y))=\mu(y)$ for all $y\in U$. Considering these kind of structures we have that  for any open set $V\subset U$ the restriction of the sections define a natural homomorphism  $\rho$ (i.e a function which conserve the relations and commute with functions) between $\mathfrak{A}(U)$ and $\mathfrak{A}(V)$:

\begin{align*}
\rho_{UV}:\mathfrak{A}(U) &\rightarrow \mathfrak{A}(V)\\
\sigma &\mapsto \sigma\upharpoonright_V.
\end{align*}

Sometimes it will be easier to define the structures $\mathfrak{A}(U)$\footnote{Here we are identifying the structure with the universe of the structure $\mathfrak{A}(U)$ this is something common in model theory to avoid some  excess of formalism.} than the complete sheaf of structures, as it will be the case of the hierarchy of Quantum Variable Sets. However  given the structures $\mathfrak{A}(U)$ for the open sets of $X$ and the morphisms $\rho_{UV}$, it is possible to reconstruct the sheaf of structures. To show how to do this, we need two important definitions.

\begin{defin}
A presheaf of structures of type $\tau$ over $X$ is an assignation $\Gamma$, such that to each open set $U \subset X$ is assigned a $\tau$-structure $\Gamma(U)=(\Gamma(U), R^{\Gamma(U)}_1,..., f^{\Gamma(U)}_1,...,c^{\Gamma(U)}_1,..)$, and if $V\subset U$ it is assigned an homomorphism $\Gamma_{UV}$ which satisfies $\Gamma_{UU}=Id_{\Gamma(U)}$ and $\Gamma_{VW}\circ\Gamma_{UV}=\Gamma_{UW}$ if $W\subseteq V\subseteq U$.
\end{defin}

It is clear that it is possible to define a presheaf of structures $\Gamma_{\mathfrak{A}}$ from a sheaf of structures assigning to each open set $U$ a $\tau$-structure which  universes are the sets $\mathfrak{A}(U)$,  and the homorphism are the $\rho_{UV}$ defined above.  Given a presheaf of structures we can construct a Sheaf of structures in the next way:

\begin{defin}
Let $\Gamma$ a presheaf of structures over $X$. Let $\mathcal{G}\Gamma$ be the sheaf of structures over $X$ such that each fiber $(\mathcal{G}\Gamma)_x$ is defined by:
\[(\mathcal{G}\Gamma)_x=\dot{\bigcup} _{U\in \mathcal{V}(x)} \Gamma(U)_{/\sim_x},\]
where $\mathcal{V}(x)$ is the set of neighbourhoods of $x$, and given $\sigma\in \Gamma(U)$ and $\lambda\in \Gamma(V)$,
\[\sigma\sim_x\lambda \Leftrightarrow \exists W\in \mathcal{V}(x), W\subseteq U\cap V\text{ such that } \Gamma_{UW}(\sigma)=\Gamma_{VW}(\lambda).\]
Let $[\sigma]_x$ the equivalence class of $\sigma$ respect to $\sim_x$. We have that:
\[([\sigma_1]_x,...,[\sigma_n]_x)\in R^{x}\Leftrightarrow \exists U\in\mathcal{V}(x)\text{ such that }(\sigma_1,...,\sigma_n)\in R^{\Gamma(U)}\]
\[f([\sigma_1]_x,...,[\sigma_n]_x)=[f^{\Gamma(U)}(\sigma_1,...,\sigma_n)]_x.\]
The space of the fibers $\bigcup_x(\mathcal{G}\Gamma)_x$ is given the topology generated by the images of the sections
\begin{align*}
a_{\sigma}:U&\rightarrow E\\
x & \mapsto [\sigma]_x.
\end{align*}
The sheaf $\mathcal{G}\Gamma$ is called the sheaf of germs of $\Gamma$.
\end{defin}

This definition can be understood using the analogy of  Galilean spacetime in the next way. If to each time interval we assign  objects such that their history or part of it develops in such interval, we can recover the instantaneous perspective  identifying two objects at each instant  if their histories coincide in an interval of time containing that instant. In the same way a relation will hold from the instantaneous point of view if it holds in a time interval containing the respective instant. This kind of contextuality will play a fundamental role in the logic that governs these models as we will see below.\\

The sheaf of germs $\mathcal{G}\Gamma_{\mathfrak{A}}$ associated to the presheaf of   sections $\Gamma_{\mathfrak{A}}$  of a sheaf $\mathfrak{A}$, is naturally isomorphic to the original sheaf. Indeed the function 
\begin{align*}
H:\bigcup_{x\in X}(\mathcal{G}\Gamma_{\mathfrak{A}})_x&\rightarrow E\\
[\sigma]_x\mapsto \sigma(x),
\end{align*}  
defines a natural isomorphism which sends in an isomorphic way each fiber $(\mathcal{G}\Gamma_{\mathfrak{A}})_x$ to the fiber $E_x$. On the other hand given a presheaf $\Gamma$ the  presheaf $\Gamma_{\mathcal{G}\Gamma}$ associated to the sheaf of germs $\mathcal{G}\Gamma$, results also isomorphic to the original presheaf just if the presheaf satisfies a further condition.

\begin{defin}
A presheaf of structures is said to be exact, if given $U=\bigcup_i U_i$ and $\sigma_i\in\Gamma(U_i)$, such that if 
\[\Gamma_{U_i,U_i\cap U_j}(\sigma_i)=\Gamma_{U_j,U_i\cap U_j}(\sigma_j) \text{ for all } i,j;\]
there exists an unique $\sigma\in\Gamma(U)$ such that $\Gamma_{UU_i}(\sigma)=\sigma_i$ for all i. And the same holds for the relations i.e if we have  some relations $R^{\Gamma(U_i)}_i, R^{\Gamma(U_j)}_j$ which are sent by the homomorphisms $\Gamma_{U_i,U_i\cap U_j}$, $\Gamma_{U_j,U_i\cap U_j}$ to a same relation for all $i,j$. There exists an unique relation $R^{\Gamma(U)}$ which is sent by the homomorphism $\Gamma_{UU_i}$ to the relation $R^{\Gamma(U_i)}_i$ for all $i$.
\end{defin}

We have then the next result.

\begin{lema}
If $\Gamma$ is an exact presheaf  then it results isomorphic to  the presheaf $\Gamma_{\mathcal{G}\Gamma}$ associated to the sheaf of germs $\mathcal{G}\Gamma$, in the sense that $\Gamma_{\mathcal{G}\Gamma}(U)\cong \Gamma(U)$ as structures and the homomorphisms $\Gamma_{UV}$ transform in the homomorphism $\Gamma_{\mathcal{G}\Gamma_{UV}}$.
\end{lema}

We will define the Quantum Hierarchy of variable sets defining an exact presheaf, the above results allow to deal indistinctly with the presheaf of structures and the associated sheaf.

\subsection{The Logic of  Sheaves of Structures}\label{logic1}

The notion of `truth" in classical physics is a contextual one. When we affirm that a property holds for a certain object at some instant,  we are referring to  a measurement realized in an extended interval of time containing that instant. The mathematical models we use to describe such situations, even if based on an absolute notion of truth, permit to capture this contextual character  using the notion of limit, which allows to define instantaneous properties in a coherent way. In some sense  the notion of limit is what makes classical logic work  in the conception of the continuum. It is surprising then,  that mathematical models which are based on an absolute notion of truth  can capture and describe physical reality, where the notion of truth is contextual, in such an effective way. Nevertheless,   it is probably this lack of contextuality what makes the classical formalism inappropriate to give a complete picture of quantum theory. In QM the notion of contextuality recovers total  new meanings related to the incompatibility of observables, and interference phenomena as we will see. Therefore, it can be fundamental to have a formalism that allows to include a notion of truth  which can contain these contextual features. This is the case of the sheaves of structures, which have a logic based on the next \textit{contextual-truth} paradigm:

\begin{center}
\textit{If a property for an extended object holds in some point of its domain then it has to hold in a neighbourhood of that point.}
\end{center}

As the objects of a sheaf of structures are the sections of the sheaf, the logic which governs them should define when a property for an extended object holds in a point of its domain of definition. The next definition contains a lot of  logical language but it is very natural.

\begin{defin} \label{logic}
Let  $L_{\tau}$ be a first order language of type $\tau$. Given a sheaf of structures  $\frak{A}$ of type $\tau$ over $X$, and a proposition 
$\varphi(v_1,...,v_n)\in L_{\tau}$ we can define by induction 
\[\frak{A}\Vdash_x
\varphi(\sigma_1,...,\sigma_n)\] 
(Which means $\frak{A}$ forces 
$\varphi(\sigma_1,...,\sigma_n)$ in $x\in X$ for  the sections
$\sigma_1,...,\sigma_n$ of $\frak{A}$ defined in $x$ or in a more elementary fashion, the property $\varphi$ holds at the node $x$ for the sections $\sigma_1,...\sigma_n$ in the sheaf $\mathfrak{A}$)  :\\

1. If $\varphi$ is an atomic formula and $t_1,...,t_k$ are $\tau$-terms: \[\frak{A}\Vdash_x
(t_1=t_2)[\sigma_1,...,\sigma_n]\Leftrightarrow
t_1^{\frak{A}_x}[\sigma_1(x),...,\sigma_n(x)]=t_2^{\frak{A}_x}[\sigma_1(x),...,\sigma_n(x)]\]
\[\frak{A}\Vdash_x
R(t_1,...,t_n)[\sigma_1,...,\sigma_n]\Leftrightarrow
(t_1^{\frak{A}_x}[\sigma_1(x),...,\sigma_n(x)],...,t_n^{\frak{A}_x}[\sigma_1(x),...,\sigma_n(x)])\in
R^{x}\]
1'. The equality between extended objects, or objects defined in function of extended objects holds at some node if and only if  the equality holds for their punctual descriptions. Analogous for the relations.\\

2. $\frak{A}\Vdash_x
(\varphi\wedge\psi)[\sigma_1,...,\sigma_n]\Leftrightarrow
\frak{A}\Vdash_x\varphi[\sigma_1,...,\sigma_n]$ and
$\frak{A}\Vdash_x \psi[\sigma_1,...,\sigma_n]$\\
2'. The conjunction of two properties holds for some extended objects at some node if and only if each property holds for those extended objects at the same node. \\

3. $\frak{A}\Vdash_x
(\varphi\vee\psi)[\sigma_1,...,\sigma_n]\Leftrightarrow
\frak{A}\Vdash_x\varphi[\sigma_1,...,\sigma_n]$ or
$\frak{A}\Vdash_x\psi[\sigma_1,...,\sigma_n]$.\\
3'. The disjunction of two properties holds for some extended objects at some node if and only if one of the properties hold for the extended objects in that node.\\

4.$\frak{A}\Vdash_x\neg\varphi[\sigma_1,...,\sigma_n]\Leftrightarrow$
exists $U$ neighbourhood of $x$ such that for all $y\in U$,
$\frak{A}\nVdash_y \varphi[\sigma_1,...,\sigma_n]$.\\
4'.The negation of a property for some extended objects hold at some node, if and only if there exists a neighbourhood of the node such that at each point in that neighbourhood the property does not hold for the extended objects.\\

5. $\frak{A}\Vdash_x
(\varphi\rightarrow\psi)[\sigma_1,...,\sigma_n]\Leftrightarrow$
Exists  $U$ neighbourhood of  $x$ such that  for all $y\in U$ if 
$\frak{A}\Vdash_y\varphi[\sigma_1,...,\sigma_n]$ then
$\frak{A}\Vdash_y\psi[\sigma_1,...,\sigma_n]$.\\

6.  $\frak{A}\Vdash_x \exists
v\varphi(v,\sigma_1,...,\sigma_n)\Leftrightarrow$ exists $\sigma$
defined in $x$ such that
$\frak{A}\Vdash_x\varphi(\sigma,\sigma_1,...,\sigma_n)$\\

7. $\frak{A}\Vdash_x\forall
v\varphi(v,\sigma_1,...,\sigma_n)\Leftrightarrow$ exists $U$
neighbourhood of $x$ such that for all $y\in U$ and all $\sigma$
define in $y$,
$\frak{A}\Vdash_y\varphi[\sigma,\sigma_1,...,\sigma_n]$.

(In 4,5,7 \quad $U$ has to satisfy $U\subseteq \bigcap_{i}dom(\sigma_i)$)

\end{defin}

Numerals 1',2',3', 4'  clarify the content of the definition without the technicalities of the logical language; 5,6, 7 are more clearly understood using  formal language. It is clear from 4,5, 7 how this is a contextual logic, however this contextuality is  better expressed in the next result which tell us that a property is verified at some node if and only if it is verified in a neighbourhood of that node.

\begin{corol} \label{localtruth}
$\frak{A}\Vdash_x \varphi[\sigma_1,...,\sigma_n]$ if and only if there exists a neighbourhood  $U$ of $x$, such that $\frak{A}\Vdash_y
\varphi[\sigma_1,...,\sigma_n]$ for all $y\in U$.
\end{corol}

Using the above definition it is possible to introduce a local semantics in a natural way, which will be more useful when dealing with presheaves of structures. Given an open subset $U\subset X$, and sections defined over $U$, we say that a proposition about these sections holds in $U$ if it holds at each point in $U$ or in other words:

\[\mathfrak{A}\Vdash_U\varphi[\sigma_1,...,\sigma_n]\Leftrightarrow \forall x\in U, \mathfrak{A}\Vdash_x \varphi[\sigma_1,...,\sigma_n]\]

This definition is determined completely by the next result. 

\begin{teor}[Kripke-Joyal semantics]\label{kripkejoyal}

$\frak{A}\Vdash_U\varphi[\sigma_1,...,\sigma_n]$ is defined by:\\
1. If $\varphi$ is an atomic formula:\\
$\frak{A}\Vdash_U\sigma_1=\sigma_2\Leftrightarrow
\sigma_1\upharpoonright_U=\sigma_2\upharpoonright_U$.\\
$\frak{A}\Vdash_U R[\sigma_1,...,\sigma_n]\Leftrightarrow
(\sigma_1,...,\sigma_n)(U)\subseteq R^{\frak{A}}$.\\
2.
$\frak{A}\Vdash_U(\varphi\wedge\psi)[\sigma_1,...,\sigma_n]\Leftrightarrow
\frak{A}\Vdash_U\varphi[\sigma_1,...,\sigma_n]$ and
$\frak{A}\Vdash_U\psi[\sigma_1,...,\sigma_n]$.\\
3.
$\frak{A}\Vdash_U(\varphi\vee\psi)[\sigma_1,...,\sigma_n]\Leftrightarrow$
there exist open sets $V,W$ such that  $U=V\cup W$,
$\frak{A}\Vdash_V\varphi[\sigma_1,...,\sigma_n]$ and
$\frak{A}\Vdash_W\psi[\sigma_1,...,\sigma_n]$.\\
4.
$\frak{A}\Vdash_U\neg\varphi[\sigma_1,...,\sigma_n]\Leftrightarrow$
For any open set  $W\subseteq
U$,\quad$W\neq\emptyset$,\quad$\frak{A}\nVdash_W\varphi[\sigma_1,...,\sigma_n]$.\\
5. $\frak{A}\Vdash_U
\varphi\rightarrow\psi[\sigma_1,...,\sigma_n]\Leftrightarrow$ for any
open set $W\subset U$,\quad if
$\frak{A}\Vdash_W\varphi[\sigma_1,...,\sigma_n]$ then
$\frak{A}\Vdash_W\psi[\sigma_1,...,\sigma_n]$.\\
6. $\frak{A}\Vdash_U\exists
v\varphi(v,\sigma_1,...,\sigma_n)\Leftrightarrow$ there exists  $\{U_i\}_i$  an open cover of  $U$ and $\mu_i$
sections defined on  $U_i$ such that 
$\frak{A}\Vdash_{U_i}\varphi[\mu_i,\sigma_1,...,\sigma_n]$ 
for all $i$.\\
7.
$\frak{A}\Vdash_U\forall v\varphi(v,\sigma_1,...,\sigma_n)\Leftrightarrow$
for any open set $W\subset U$ and $\mu$ defined on  $W$,
$\frak{A}\Vdash_W\varphi(\mu,\sigma_1,...,\sigma_n)$.
\end{teor}

The logic just defined can be seen as a multivalued logic with truth values that variate over the Heyting algebra of the open sets of the base space $X$. Let $\sigma_1,..., \sigma_n$ sections of a sheaf $\mathfrak{A}$ defined over an open set $U$, we define the \textit{``truth value"} of a proposition $\varphi$ in $U$ as:
\[ [[\varphi(\sigma_1,...,\sigma_n)]]_{U}:=\{x\in U: \mathfrak{A}\Vdash _x \varphi [\sigma_1,...,\sigma_n]\}\]
From corollary \ref{localtruth} we know that $[[\varphi(\sigma_1,...,\sigma_n)]]_{U}$ is an open set, thus we can define a valuation as a topological valuation on formulas:
\[T_{U}:\varphi\mapsto [[\varphi(\sigma_1,...,\sigma_n)]]_{U}.\] 
The definition of the logic allows to define the value of the logic operators in terms of the operations of the algebra of open sets. For instance, the proposition $\neg\varphi$ is valid in a point if there exists a neighbourhood of that point where $\varphi$ does not hold at each point of the neighbourhood, then $\neg\varphi$ holds in the interior of the complement of the set where $\varphi$ holds i.e $[[\neg \varphi]]_{U}= Int (U\setminus [[\varphi]]_{U}),$. Reasoning in analogous way we have:

\begin{enumerate}
\item $[[\sigma_1=\sigma_2]]_{U}=\{x\in U:\sigma_1(x)=\sigma_2(x)\}$.
\item $[[R[\sigma_1,...,\sigma_n]]]_{U}=\{x\in U: (\sigma_1(x),...,\sigma_n(x))\in R^{\mathfrak{A}}\}$
\item $[[\neg \varphi]]_{U}= Int ((U\setminus [[\varphi]]_{U}),$
\item $[[\varphi\wedge \psi]]_{U}=[[\varphi]]_{U}\cap[[\psi]]_{U},$
\item $[[\varphi\vee \psi]]_{U}=[[\varphi]]_{U}\cup [[\psi]]_{U},$
\item $[[\varphi\rightarrow \psi]]_{U}= Int((U\setminus[[\varphi]]_{U})\cup [[\psi]]_{U}),$
\item $[[\exists u \varphi (u)]]_{U}=\bigcup_{\sigma\in \mathfrak{A}(W),  W\subset U}[[\varphi(\sigma)]]_{W},$
\item $[[\forall u\varphi(u)]]_{U}=Int(\bigcap_{\sigma\in \mathfrak{A}(W), W\subset U}[[\varphi(\sigma)]]_{W}).$
\end{enumerate}

It follows then 

\[\mathfrak{A}\Vdash_{U} \varphi[\sigma_1,...,\sigma_n]\leftrightarrow [[\varphi[\sigma_1,...,\sigma_n]]]_{U}=U.\]
 
Therefore, since the open sets form a complete Heyting algebra, and we have defined the logic operators in terms of the algebra operations, we have that the formulas that get the value 1 in a complete Heyting algebras are forced in every node of the sheaf of structures. This proves that the laws of the intuitionistic logic, which are those that assume the value 1 in a Heyting algebra, are forced in each node of a Sheaf of Structures.\\
  
In this way each sheaf of structures is ruled by a logic intermediate between intuitionistic logic and classical logic. Intermediate because the geometry of the spaces which determine the sheaf will determine how close or far of each logic is the intrinsic logic of the sheaf (see next section).  This connection between logic and geometry is a very interesting fact that can be important in a future theory of Quantum Gravity.

\subsection{Excursus: Logic, Physics and Geometry}

In recent years the use of topoi has been suggested  to construct a kind of Quantum Geometry  to unify GR and QM (see for example \cite{crane}), based on the conjectured  inadequacy of the classical concept of manifold to construct a theory of Quantum Gravity. The conjecture tacitly contained there is the necessity  of an intuitionistic version of  spacetime. If an intuitionistic model can capture the essence of quantum mechanics, as I am proposing in this work, it will be natural to think that to unify GR and QM,  we will need an intuitionistic version of spacetime. I want to show how probably within General Relativity, using the notion of Sheaf of Structures, we can  also find strong motivations to think that this is the right route to follow.\\
\indent The connections between Logic and Geometry included within the formalism I am presenting here (and also in the topos formalism), were one of the main motivations to think that probably these methods can be useful in a future theory of Quantum Gravity. As we will see below the logic structure of quantum mechanics in  this context derives in its interpretation. On the other hand, these tools allow to obtain a new description of the continuum which arise from  the logic structure of QM. The continuum is what we use in geometry to measure, the notion of Lorentzian manifold is an expression of the continuum to capture the essence of spacetime. If we can construct a definition of manifold which arise from the logic of QM, we will probably get automatically a Quantum Gravity theory. In other words, in some sense Quantum Mechanics represents Logic, and General Relativity is in essence Geometry; if we want to connect these two, the sheaf of structures formalism (or the Topoi formalism) seems to be an appropriate route to follow.\\

To understand better what I mean by  connections between logic and geometry lets consider a simple example.  Consider $E=\mathbb{R}\cup\{q\}$ and $X=\mathbb{R}$, where $X$ has the usual topology and $E$ has the usual topology of $\mathbb{R}$ plus the open neighbourhoods of $q$ which are of the form $(U\setminus \{0\})\cup\{a\}$ where $U$ is a neighbourhood of $0$. 

\begin{figure}
\centering
\includegraphics[scale=0.40]{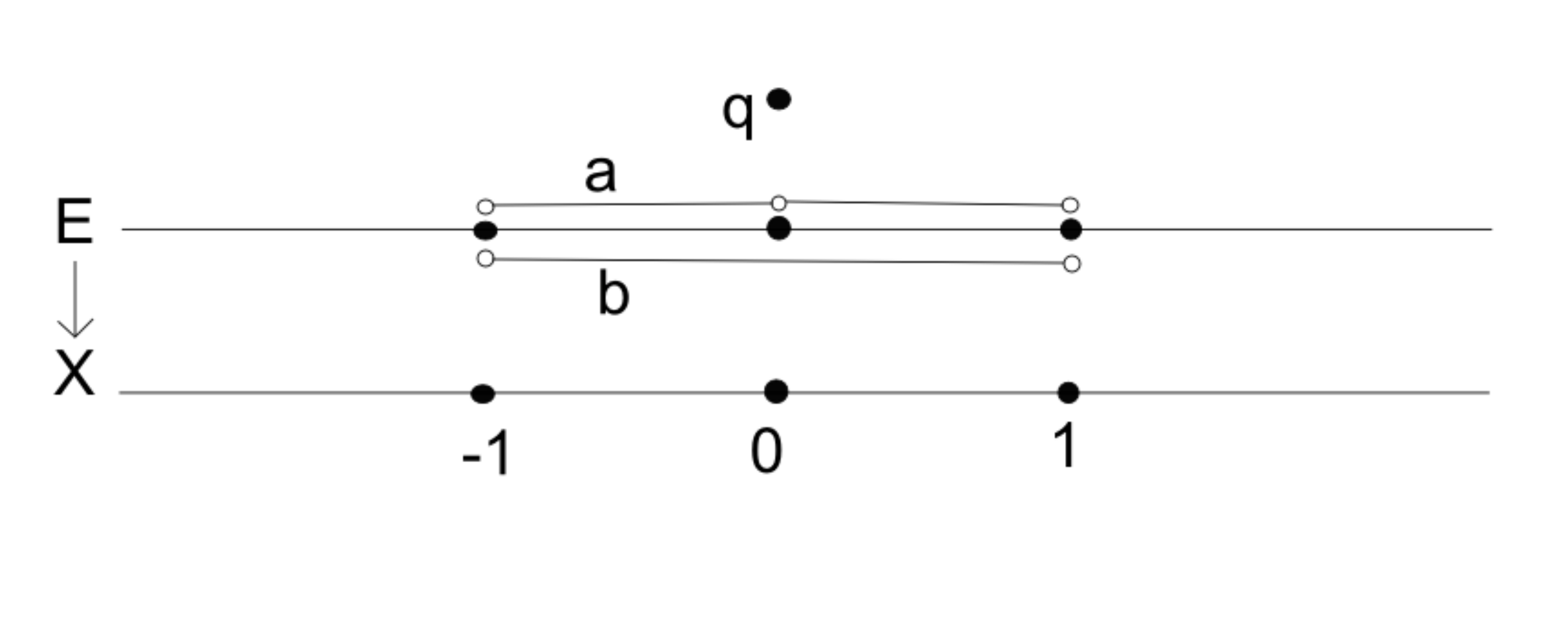}
\caption{The Hausdorff property in E determines the validity of the excluded middle property. A geometric property determines a logic property.}
\label{hausdorff}
\end{figure}

From $E$ and $X$ we can construct a natural sheaf  with $p:E\rightarrow X$ such that $p(x)=x$ if $x\neq q$ and $p(q)=0$ (see Fig \ref{hausdorff}). Consider the sections $a:(-1,1)\rightarrow E$, $b:(-1,1)\rightarrow E$ such that $a(x)=x$ if $x\neq 0$ and $a(0)=q$ and $b(x)=x$ for all $x$. Even if $a$ and $b$ are two different sections, from the perspective of $0$  the sections are either different or equal. In other words on the node $0$ we have:
\[ \nVdash_0 (a=b\vee \neg(a=b)).\]
Indeed it is clear that $\nVdash_0 a=b$ (i.e  $\Vdash_0 a=b$ does not hold) because  $a(0)\neq b(0)$ (see 1 in the definition \ref{logic}). On the other hand  $\nVdash_0 \neg(a=b)$ because each neighbourhood of $0$ contains points $x$ such that $a(x)=b(x)$. However, note that if $x\in (-1,1)\setminus \{0\}$, $\Vdash_x (a=b)$.  The non validity of the excluded middle in $0$ depends on the fact that the Hausdorff  property is not valid on the fiber over $0$. It can be proved for an arbitrary sheaf the next result:
\[\Vdash_x \forall z,y(z=y\vee\neg(z=y))\Leftrightarrow \exists U\in \mathcal{V}(x) \text{ s.t. } p^{-1}(U) \text{ is a Hausdorff space }.\] 
Thus, this example shows how the logic of the sheaf of structures is related with the geometric properties of the sheaf.\\

In General Relativity, singularities in  a  spacetime $M$ can be added as boundary points in an extended manifold $M^{+}=M\cup\partial M$ in such a way that the differential structure of the extended manifold captures the structure of the original manifold in some way. In some of these boundary constructions as those developed by Geroch and Schmidt (see \cite{hawking}  p.217 . p.276 respectively) some points, where the Hausdorff property does not hold, appear as a result of the construction of the boundary. Furthermore this situation, in the Schmidt construction, occurs  in the  Schwarzschild and Friedman solutions \cite{johnson}, which are our paradigms of strong curvature singularities. These non-Hausdorff features have been seen as defects of these boundary constructions, but within the context of the example exposed above, probably what these results are telling us is that the singularities are a result of the breakdown of the logic which governed these models. The breakdown of the Hausdorff property is probably the breakdown of the excluded middle in the singularity point in some appropriate sheaf description of spacetime. The difficulty to construct a satisfactory boundary definition such that the extended manifold conserve all the desirable features of a classical spacetime,  is probably due to the fact that this formalism is inadequate to capture these situations. The breakdown of the Hausdorff property is a sign of the breakdown of any possibility of classical description, non- Hausdorff spacetimes affect the formulation of the Cauchy problem,  in non-Hausdorff spacetime we can find bifurcating geodesics and classical determinism can loose all its sense. A new conception of manifold is needed, but within this new conception, local determinism will probably have a new meaning based on the ontology of QM, this is why a theory of Quantum Gravity  cannot be constructed before settling what Quantum Reality is.\\

\section{Quantum Set Theory}

Now we want to construct a mathematical universe which is founded over the logic intrinsic to QM. As we saw above the logic of a sheaf of structures  with base space $X$ can be seen as a multivalued logic which take its values over the algebra of open subsets of $X$. Therefore, if we construct sheaves  over a topological space which algebra of open sets  capture the logic of QM,  we will find automatically models ruled by this logic.\\
\indent The mathematical models that  have been used so far in physics are founded over classical logic and classical set theory. If we want an alternative mathematical universe founded over the logic of QM, it will be important to find one as a natural generalization of the classical universe. Thus it will be important for this model  also to be a kind of model of set theory where all the classical constructions find a natural counterpart.\\
The classical mathematical universe is the Von Neumann Hierarchy, which is constructed inductively in the next way. Let
 \[ V_0=\emptyset\]
\[V_{\alpha+1 }=\mathcal{P}(V_\alpha)\]
\[V_\lambda=\bigcup_{\alpha<\lambda}V_{\alpha} \text{ If }\lambda\text{ is a limit ordinal,}\]
where $\mathcal{P}(V_{\alpha})$ is the power set of $V_{\alpha}$ (i.e the set of all the subsets of $V_{\alpha}$). The set 
\[ \mathbb{V}=\bigcup_{\alpha\in On} V_{\alpha},\]
where $On$ is the class of all ordinals, is the Von Neumann Hierarchy of classical sets. Probably the notion of ordinal can be new for someone working in physics, but it is not fundamental to understand the model, just think ordinals as a generalization of natural numbers which cover all possible infinities. $\mathbb{V}$ with the classical belonging relation, $\in$, between sets,  is the universe of classical mathematics,  a model (in the sense of classical logic) of the axioms of ZFC (Zermelo-Fraenkel+Choice). We will construct an analogous model which extend the notion of set, as a variable object  over a sheaf, where the $\in$ relation recover a new contextual meaning.

\subsection{The Cumulative Hierarchy of Variable Sets} 

To introduce the definition of the cumulative hierarchy  we can use the interpretation that categorists use to introduce  the objects of $SET^{\mathbb{P}}$, the advantage here is that the motivation translates literally in the definition; we obtain a truly structure of variable of sets, where using the classical  belonging relation between sets it is possible to define an analogous  belonging relation between variable sets and not as arrows between objects. The definition I  present below  was  originally introduced in \cite{caicedo} and I used  it in \cite{benavides},\cite{benavides3} to get a new proof of the independence of the Continuum Hypothesis. In those articles you can find a more detailed description, here I present a brief introduction enough to follow the main ideas contained below.\\

Using the comprehension axiom of set theory, given a proposition $\varphi(x)$, for any set $A$ we can construct a set $B$ such that $x\in B$ if and only if $x\in A$ and $\varphi(x)$ is ``truth" for $x$ or in other words,
\[B=\{x\in A: \varphi(x)\}.\]
As we have seen  in the sheaves of structures the notion of truth is not absolute but it is based on a contextual truth paradigm. To see how this notion of truth translates when we talk about sets lets consider a particular example. Let $\varphi(x)$ be the next proposition \textit{`` x is an even number greater or equal than 4 and x can be written as the sum of two prime numbers"} . In this moment at my ubication on spacetime $p:=$\textit{``between the 21:00 and the 22:00, of the 5 March 2011, in some place in London"} I cannot assure that $\{z\in\mathbb{N}:\varphi(x)\}=\{x\in \mathbb{N}: x\geq 4 \wedge (x$ is even $)\}$, i.e. that the Goldbach Conjecture is true. However $\varphi$ can be used to define a set at the node $p$,
\[\varphi(p)=\{x:\varphi(x) \text{ holds at the node } p.\}\] 
The Goldbach conjecture has been verified for a huge number of even numbers, and this set of numbers will keep growing each time we will find a new prime number. Probably (maybe not) some day the Goldbach Conjecture will be proved or disproved, but the important feature that follow from this example is that instead of conceiving sets as absolute entities, we can conceive them as variable structures which variate over our \textit{Library of the states of knowledge}. It is natural then to conceive the set of nodes where our states of Knowledge variates as nodes in a partial order or points in a topological space, that can represent, for instance,  the causal structure of spacetime. Our "states of Knowledge" will be then structures that represent the sets as we see them in our nodes. Therefore, from each node we will see arise a cumulative Hierarchy of variable sets, which structure will be conditioned by the perception of the variable structures in the other nodes that relate to it. Or more precisely.
  
\begin{defin}\label{hierarchy}
Let $X$ be an arbitrary topological space, the cumulative hierarchy of variable sets over $X$ is defined in the next way. Given $U\in Op(X)$\footnote{ $Op(X)$ denote the non empty open sets of $X$} we define inductively:
\begin{align*}
V_0(U)=& \emptyset\\
V_{\alpha+1}(U)=&\{ f:Op(U)\rightarrow \bigcup_{W\subseteq U}\mathcal{P}(V_{\alpha}(W)): 1. \text{ If }W\subseteq U \text{ then } f(W)\subseteq V_{\alpha}(W),\\
& 2. \text{ If }V\subseteq W\subseteq U,\text{ then for all }g\in f(W),\quad g\upharpoonright_{Op(V)}\in f(V),\\
& 3. \text{ Given }\{U_i\}_i \text{ an open cover of } U \text{ and } g_i\in f(U_i)\\
& \text{such that }  g_i\upharpoonright_{op(U_i\cap U_j)}=g_j\upharpoonright_{op(U_i\cap U_j)} \text{ for any }i,j,\\
& \text{there exists }g\in f(U) \text{ such that }g\upharpoonright_{op(U_i)}=g_i\text{ for all } i\}\\
V_{\lambda}(U)=&\bigcup_{\alpha<\lambda}V_{\alpha}(U) \text{ if } \lambda\text{ is a limit ordinal,}\\
V(U)=&\bigcup_{\alpha\in On}V_{\alpha}(U).
\end{align*}
The valuation $V$ over the open sets constitute an exact  presheaf of structures, which Sheaf of Germs $\mathbb{V}^{X}$  constitute the cumulative hierarchy of variable sets. 
\end{defin}  

For each $U\in Op(X)$ the set $V(U)$ is a set of functions defined over $Op(U)$ which values for $W\in Op(U)$ are functions over $Op(W)$ which values for $V\in Op(W)$ are functions over $Op(V)$ and so on.\\
\indent  The next step is to define the $\in$ relation. In a natural way we define:

\[\Vdash_U f\in g \Leftrightarrow f\in g(U),\]

which means that respect to the context $U$,  $f$ belongs to $g$ if and only if $f\in g(U)$ as classical sets\footnote{Compare this definition with the definition of the $\in$ relation on Boolean valued models}. At first sight this construction seems to be odd but in the results presented below we will see that the constructions are extremely natural and extremely powerful. The first result is that, as we expect, this model is a natural generalization of the universe of classical sets, a model where the axioms of Set Theory are satisfied respect to the logic of  sheaves of structures \cite{villa}\cite{benavides}\cite{benavides3}.

\begin{teor}
For any topological space $X$,
\[\mathbb{V}^{X}\Vdash ZF.\]
In other words for any $x\in X$ (or  $U\in Op(X)$) , at the node $x$ (at the context $U$) the axioms of set theory are satisfied i.e $\Vdash_x ZF$ ( $\Vdash_U ZF$).
\end{teor}
 
To see what this means lets see how the axiom of comprehension cited above is valid in this more general context.\\

\textbf{Axiom of Comprehension:}\textit{ Let $\varphi(x,z,w_1,...,w_n)$ be a formula. For every set $z$ there exists a set $y$ such that, $x\in y$ if and only if $x\in z$ and $\varphi(x)$ holds for $x$ }

\begin{lema} Let $X$ be a topological space, in $\mathbb{V}^{X}$ for any $U\in Op(X)$,
\[\Vdash_U \forall z\forall w_1,...,w_n \exists y \forall x(x\in y \leftrightarrow x\in z \wedge \varphi)\]
\end{lema}

\begin{proof}
Following the logic of  sheaf of structures, contained in theorem  \ref{kripkejoyal}, to prove the above result means to prove that  for all $W\subseteq U$ and $z,w_1,...,w_n\in V(W)$ there exists an open cover $\{W_i\}_i$ of $W$ and $y_i\in V(W_i)$ such that for any $T_i\subseteq W_i$ and $x\in V(T_i)$ for $Y_i\subseteq T_i$, $x\in y_i(Y_i)$ if and only if\footnote{To be rigorous in this last expression we have to write $x\upharpoonright_{Y_i}\in y_i\upharpoonright_{Y_i}(Y_i)$ however we will omit these technicalities otherwise the notation will become very confusing.}  $x\in z(Y_i)$ and $\Vdash_{Y_i}\varphi(x)$. It may seem that proving a result in this new logic is very complicated but as we will see, this and the next results have elementary proofs. 
Let $z, w_1,...,w_n \in V(U)$ and 
\[y:Op(U)\rightarrow  \bigcup_{W\subseteq U} V(W)\]
\[y(W)=\{x\in z(W): \Vdash_W \varphi(x)\}.\]
Note that we construct $y(W)$ using the axiom comprehension in the classical universe. It is clear that the variable set $y$ just defined satisfies the result described above taking the appropriate restrictions on the respective sets there cited. Note for example that as an open cover of $W\subseteq U$ we can take the cover formed  by the set $W$, in other results below we will need the definition of the existential using a non trivial open cover but here it is not the case. However we need to prove that the $y$ just defined is an element of the structure of variable sets i.e we need to verify properties 1-3 in definition \ref{hierarchy}. Let  $\alpha$ the minimum ordinal such that $z\in V_{\alpha+1}(U)$. Since $y(W)\subset z(W)$ then $y(W)\subset V_{\alpha}(W)$, therefore 
\[y:Op(U)\rightarrow \bigcup_{W\subseteq U} P(V_{\alpha}(W)).\]
Given $T\subseteq W\subseteq U$ and $g\in y(W)$ we have $g\in z(W)$ and $\Vdash_W \varphi(g)$. Thus $g\upharpoonright_{Op(T)}\in z(T)$ and $\Vdash_T \varphi(g\upharpoonright_{Op(T)})$, then $g\upharpoonright_{Op(T)}\in y(T)$.  Property 3 in the definition follows in analogous way.  Therefore $y\in V(U)$. 
\end{proof}

Even if generally $V(U)$ and the classical Von Neumann hierarchy $\mathbb{V}$ are not isomorphic for any open set $U$ there is an standard way to embed $\mathbb{V}$ in $V(U)$. Given an arbitrary set $a\in \mathbb{V}$ let 

\[\widehat{a}(U):Op(U)\rightarrow\bigcup_{W\subseteq U}V(W)\]
\[\widehat{a}(U)(W)=\{\widehat{b}(U)\upharpoonright_{Op(W)}:b\in
a\}.\]

we have that  $\widehat{a}(U)\upharpoonright_{Op(W)}=\widehat{a}(W)$ and it can be proved the next result \cite{benavides3}.

\begin{teor} Let $U\subseteq X$ an open subset, then 
 \begin{align*}
 \Psi:\mathbb{V}&\rightarrow V(U)\\
a &\mapsto \widehat{a}(U)
\end{align*}
is a monomorphism. 
\end{teor}

Using this embedding it can be easily proved that the set forced as the empty set on each open set $U$ is precisely $\widehat{\emptyset}(U)$ or in other words it holds
\[\Vdash_U \forall y (y\notin \widehat{\emptyset}(U)).\]
 In the same way it can be proved that the set  forced as the set of natural numbers (i.e. the minimum inductive set that contains the empty set) is precisely $\widehat{\mathbb{N}}(U)$ \cite{benavides3}. The sets forced as the integers and the rational numbers are also $\widehat{\mathbb{Z}}(U)$ and $\widehat{\mathbb{Q}}(U)$ respectively.  However the set  forced as the real numbers generally not coincide with $\widehat{\mathbb{R}}(U)$, for example as we will see below,  this is the case of the set that is forced as the real numbers in the case related to a base space which captures the essence of quantum logic.\\ 

Now we have a way to construct new mathematical universes over arbitrary topological spaces. We need then a topological space able to capture the essence of Quantum Logic to obtain a mathematical Quantum Universe able to capture Quantum Reality.

\subsection{The Cumulative Hierarchy of Quantum Variable Sets}\label{QVS}

As we saw above one  important property of the Logic of sheaves of structures is the possibility to define a contextual notion of truth. This contextual truth paradigm is not just intrinsic to the notion of truth in classical physical reality but within the quantum realm it can recover new fundamental meanings.\\
\indent  One first aspect of contextuality arises from the Kochen-Specker theorem \footnote{See \cite{isham4} chapter 9  for a nice discussion about the Kochen-Specker theorem.}. In a few words what the Kochen-Specker theorem says is that in QM does not  exist a valuation function from the set of physical observables, intended as self adjoint operators on a Hilbert space $H$, if the dimension of $H$ is greater than two. A valuation is a function $\lambda$ from  the set of self-adjoint operators $B_{sa}(H)$ on $H$ to the real numbers, $\lambda:B_{sa}(H)\rightarrow \mathbb{R}$, which satisfies:
\[1.\lambda(A) \text{ belongs to the spectrum of }A\in B_{sa}(H)\]
\[2.\lambda(B)=f(\lambda(A))\text{ whenever } B=f(A)\]
with $f:\mathbb{R}\rightarrow \mathbb{R}$  a Borel function.\\
\indent On the other hand, locally or in a contextual sense, there exist valuations. If we consider an abelian Von Neumann subalgebra  $\mathcal{U}$ of the algebra of operators,  the elements of the Gelfand spectrum of $\mathcal{U}$ (i.e the positive linear functions $\sigma:\mathcal{U}\rightarrow\mathbb{C}$ of norm 1 such that $\sigma(AB)=\sigma(A)\sigma(B)$ for all $A,B \in \mathcal{U}$) when restricted to the self-adjoint operators in $\mathcal{U}$ constitute a set of valuations which satisfy 1 and 2 above (see \cite{kadison}) . What the Kochen-Specker theorem is telling us is that none of this local valuations can be extended globally.\\
\indent A history is characterized by the values that physical variables take on it; therefore, we can see these valuations as histories, the Kochen-Specker theorem is telling us that the space of histories has a non trivial contextual structure. Thus, a first notion of context is given  by the histories associated to an abelian Von Neumann subalgebra of the algebra of operators of a Hilbert space.  As the elements of the Gelfand spectrum of an algebra can be extended in different ways in two non compatible extensions of the original algebra, this contextual character is telling us that what is perceived as true by a history is conditioned by the context of similar histories, where similar is intended in this case as histories that can be expressed in function of the same observables.\\
 
A second notion of contextuality arises from the phenomenon of interference in Quantum Mechanics. Interference as intended in the Deutsch-Everett multiversal interpretation\footnote{See \cite{deutsch} chapter 2, \cite{deutsch4} chapter 11 or \cite{deutsch2} for a non technical explanation of the Quantum Multiverse and \cite{deutsch3} for a more technical approach. } of Quantum Mechanics is the way as two histories can affect each other. Interference phenomenons are strong enough to be detected only between universes or histories that are very similar, thus what is perceived by a history depends on the context of histories close to it, where close is intended in the sense of being similar. To describe this second contextual feature we use projections operators on a Hilbert space $H$. Let $\mathcal{U}$ be an abelian Von Neumann algebra of operators, $S_\mathcal{U}$ the Gelfand spectrum  of $\mathcal{U}$ and $P(\mathcal{U})$ the set of projections of $\mathcal{U}$. Given $P'\in P(\mathcal{U})$ and $\sigma\in S_\mathcal{U}$ we have
\[\sigma(P')=\sigma(P'^{2})=\sigma(P')\sigma(P')\]
then $\sigma(P')\in \{0,1\}$. Consider a self-adjoint operator $A\in \mathcal{U}$ such that 
\[A=\sum_{n=1}^{N}a_n P'_n\]
is the spectral representation of $A$ in $\mathcal{U}$. Let  $\lambda\in S_\mathcal{U}$  such that $\lambda(A)=a_m$, then $\lambda(P'_m)=1$. Thus if $A$ represents a physical observable, we have that in all the histories $\lambda$  such that $\lambda(P'_m)=1$ the physical observable $A$ assumes the value $a_m$. Therefore, given a proposition $P'\in P(\mathcal{U})$ the set 
\[P=\{ \lambda \in S_\mathcal{U}: \lambda(P')=1\}\] 
is a context of histories which are similar in the sense that some physical observables assume the same values or the values satisfy the same inequalities in each history.\\

The first notion of contextuality has played an important role in the development of the Consistent Histories program (see \cite{griffiths}) and also it is the base of the approach based  on topos theory developed by Isham and others (see \cite{isham} and \cite{doering}). The importance of the second notion has not been  considered deeply but, as we will see below, this notion is fundamental and probably will explain how to relate Classical and Quantum realities.  I conjecture that these two notions of contextuality capture the essence of the logic of classical Quantum Mechanics. Therefore, to develop a formalism  able to capture the essence of Quantum Reality  using the tools developed above, we need a base space $X$ which structure captures these two notions of contextuality. Unfortunately to capture both notions we will probably need something more general than a topological space and the respective generalization of the tools developed above, and it is here that tools of topos theory will be probably needed, this will be developed in  a future work. However,  we can limit ourselves to a context $X=S_\mathcal{U}$ where $\mathcal{U}$ is an abelian Von Neumann algebra, and on $X$ we can consider the topology which open sets are the sets $P$ associated to a projector operator $P'\in P(\mathcal{U})$ as defined above. This will be enough to capture the importance of the second notion of contextuality. The topology $\{P\}_{P'\in P(\mathcal{U})}$ is a topology where all the open sets are clopen and where the algebra of  open sets is a boolean algebra isomorphic to the boolean algebra of projector operators $P(\mathcal{U})$. The Cumulative hierarchy of variable sets constructed over the topological space $\langle X=S_\mathcal{U}, \{P\}_{P'\in P(\mathcal{U})}\rangle$ where $\mathcal{U}$ is an abelian Von Neumann Algebra is what we will call  \textit{The Cumulative Hierarchy of Quantum Variable Sets.} The objects of this model will be the sections of the sheaf $\mathbb{V}^{X}$, which result to be extended objects that variate over the space of histories or universes $X=S_{\mathcal{U}}$, in a few words multiversal objects. This characteristic will probably be ideal to describe quantum particles, to put it in D. Deutsch words:
\begin{center}
\textit{Thanks to the strong internal interference that it is continuously undergoing, a typical electron is an irreducibly multiversal object, and not a collection of parallel-universe or parallel-histories objects. \footnote{ See \cite{deutsch4} page 291.}}
\end{center}

We will see that in this new mathematical universe the tools we use to understand Quantum Mechanics recover new fundamental meanings that can be fundamental to settle a definite  interpretation of QM. Indeed, once we construct the continuum in this mathematical universe we find a fundamental result which shows the interpretative power of this model respect to Quantum theory.

\subsection{The Quantum Continuum}

As in the previous section consider $\mathcal{U}$ an abelian Von Neumann subalgebra of the algebra of operators of a Hilbert space $H$, and $X=S_\mathcal{U}$ the Gelfand spectrum of $\mathcal{U}$ with the topology given by the sets $\{P\}_{P'\in P(\mathcal{U})}$ as defined above. We want to construct the continuum, or in other words, the set which is forced as the real numbers on the hierarchy of variable sets over $X$. To construct this set  we will use the definition of the real numbers given by  Dedekind cuts, therefore at each open (clopen) set  $P$  the object $\mathbb{R}(P)$ will be  the set that is forced as the set of Dedekind cuts i.e 

\[\Vdash_{P} \mathbb{R}(P)=\{(L,U)\in \mathcal{P}(\widehat{\mathbb{Q}}(P))\times \mathcal{P}(\widehat{\mathbb{Q}}(P)): (L,U) \text{ is a Dedekind cut }\};\]

where being a Dedekind cut means:

\begin{enumerate}
\item $\Vdash_{P} \exists q\in \widehat{\mathbb{Q}}(P)(q\in L)\wedge \exists r\in \widehat{\mathbb{Q}}(p)(r\in U).$
\item $\Vdash_{P} \forall q,r\in\widehat{\mathbb{Q}}(P)(q<r\wedge r\in L \rightarrow q\in L).$\\
$\Vdash_{P} \forall q,r\in \widehat{\mathbb{Q}}(P)(r<q\wedge  r\in U\rightarrow q\in U).$
\item  $\Vdash_P q\in\widehat{\mathbb{Q}}(P) (q\in L\rightarrow \exists r\in\widehat{\mathbb{Q}}(P)((r\in L\wedge q<r)).$\\
$\Vdash_P q\in\widehat{\mathbb{Q}}(P) (q\in U\rightarrow \exists r\in\widehat{\mathbb{Q}}(P)((r\in U\wedge q>r)).$
\item $\Vdash_P \forall q,r\in\mathbb{Q}(q<r\rightarrow (q\in L \vee r\in U)).$
\item $\Vdash_P L\cap U=\emptyset $.
\end{enumerate}

Since the set $\widehat{\mathbb{Q}}(P)$ is the embedding of $\mathbb{Q}$ in $V(P)$,  and the construction of the order relation does not introduce anything new, the order relation on $\widehat{\mathbb{Q}}(P)$ is equivalent to the order relation in $\mathbb{Q}$, i.e. 
\[\Vdash_P \widehat{q}(P),\widehat{r}(P)\in\widehat{\mathbb{Q}}(P)\wedge \widehat{q}(P)<\widehat{r}(P)\quad \text{if and only if}\quad q<r \text{ in the classical sense. }\]
Using \ref{kripkejoyal} the above conditions mean:

\begin{enumerate}

\item i) There exist $q\in \widehat{\mathbb{Q}}(P)(P)$ such that $q\in L(P)$ (Analogously for $U$) or\\
ii) There exist an open cover $\{P_i\}_{i\in I}$ of $P$ such that there exists $q\in \widehat{\mathbb{Q}}(P_i)(P_i)$ and $q\in L(P_i)$. (Analogously for $U$)
\item For all the open sets $Q\subseteq P$ if  $T\subseteq Q$ is open, $\widehat{q}(T),\widehat{r}(T)\in \widehat{\mathbb{Q}}(T)(T)$, $q<r$  and $\widehat{r}(T)\in L(T)$ then $\widehat{q}(T)\in L(T)$. (Analogously for $U$).
\item For all the open sets $Q\subseteq P$ if $\widehat{q}(Q)\in \widehat{\mathbb{Q}}(Q)(Q)$ then there exists an open cover $\{Q_i\}$ of $Q$ and  $\widehat{r}(Q_i)\in \widehat{\mathbb{Q}}(Q_i)(Q_i)$ such that $r> q$ and $\widehat{r}(Q_i)\in L(Q_i)$. (Analogously  for $U$).
\item  For all the open sets $Q\subseteq P$,  given $\widehat{q}(Q), \widehat{r}(Q)\in \widehat{\mathbb{Q}}(Q)(Q)$ if $T\subseteq Q$,  $q<r$  then  there exists open sets $T_1$, $T_2$ such that $T=T_1\cup T_2$, $\Vdash_{T_1}\widehat{q}(T_1)\in L$ and $\Vdash_{T_2}\widehat{r}(T_2)\in U$.
\item For all the open sets $Q\subseteq P$ and $\widehat{q}(Q)\in\widehat{\mathbb{Q}}(Q)(Q)$ not both $\widehat{q}(Q)\in L(Q)$ and $\widehat{q}(Q)\in U(Q)$

\end{enumerate}

Given a self adjoint operator $A\in \mathcal{U}$, let  $\{P'_r\}_{r\in\mathbb{R}}$ be the spectral family of operators associated to $A$ (see \cite{kadison}). The family $\{P'_r\}_{r\in\mathbb{R}}$ is a family of operators contained in $P(\mathcal{U})$ which satisfy:
\begin{enumerate}
\item $P'_q\wedge P'_r=P'_{min\{r,s\}}$,
\item $\bigwedge_{r\in\mathbb{R}} P'_r=0$,
\item $\bigvee_{r\in\mathbb{R}} P'_r=1$,
\item $\bigwedge_{q\leq r}P'_r=P'_q$ for every $q\in\mathbb{R}$
\end{enumerate} 

Thus the above spectral family defines a family of clopen subsets, $\{P_r\}_{r\in\mathbb{R}}$ which satisfy  the analogous properties:

\begin{enumerate}
\item $P_q \cap P_r=P_{min\{r,s\}}$,
\item $\bigcap_{r\in\mathbb{R}} P_r=\emptyset$,
\item $\bigcup_{r\in\mathbb{R}} P_r=X$,
\item $\bigcap_{q\leq r}P_r=P_q$ for every $q\in\mathbb{R}$
\end{enumerate} 

Using this family of clopen subsets for $P\in Op(X)$ we define:

\begin{align*}
U_A(P):Op(P)& \rightarrow \bigcup_{Q\subseteq  P}V(Q)\\
Q & \mapsto \{\widehat{q}(Q)\in \widehat{\mathbb{Q}}(Q)(Q): \exists r\in\mathbb{Q}, r<q,  Q\nsubseteq P_r^{c} \}
\end{align*}

\begin{align*}
L_A(P):Op(P)& \rightarrow \bigcup_{Q\subseteq P}V(Q)\\
Q & \mapsto \{\widehat{q}(Q)\in \widehat{\mathbb{Q}}(Q)(Q): Q\subseteq P_q^{c}\}
\end{align*}

where $P^{c}_q$ denotes the complement of $P_q$ in $X$.  We want to show that $U_A$ and $L_A$ define a real number\footnote{We will use the notation $U_A, L_A$ instead of $U_A(P), L_A(P)$, here again we have $U_A(P)\upharpoonright_{Op(Q)}=U_A(Q)$ if $Q\subseteq P$, then we can use this notation without generating confusion}, proving that the conditions  1-5 above  hold:

\begin{enumerate}

\item a-) $\{P\cap P_q^{c}\}_{q\in\mathbb{Q}}$ is an open cover of $P$, indeed,

\begin{align*}
\bigcup_{q\in \mathbb{Q}} (P\cap P^{c}_q)&=P\cap (\bigcup_{q\in\mathbb{Q}} P^{c}_q)\\
&=P\cap (\bigcap_{q\in\mathbb{Q}} P_q)^{c}\\
&=P. 
\end{align*}

We have $P\cap P^{c}_q\subseteq P^{c}_q$ thus $\widehat{q}(P\cap P^{c}_q)\in L_A(W\cap P^{c}_q)$.\\
b-) Since $P\neq \emptyset$, there exists $r\in\mathbb{Q}$ such that $P\nsubseteq P^{c}_ r$, otherwise $P\subseteq \bigcap_{q\in\mathbb{Q}} P^{c}_q=\emptyset$. Therefore, taking $q>r$ by definition we have $\widehat{q}(P)\in U_A(P)$.

\item a-) Let $Q\subseteq P$ be an open set and $\widehat{q}(Q),\widehat{r}(Q)\in \widehat{\mathbb{Q}}(Q)(Q)$. Suppose $r<q$ and $\widehat{q}(Q)\in L_A(Q)$, then $Q\subseteq P^{c}_q$.  Since $P_r\cap P_q=P_r$ we have $P_r\subset P_q$ and $P^{c}_q\subset P^{c}_r$. Thus $Q\subset P^{c}_r$ which implies $\widehat{r}(Q)\in L_A(Q)$.\\
b-) On the other hand if we suppose $r<q$ and $\widehat{r}(Q)\in U_A(Q)$, there exists $s\in\mathbb{Q}$, $s<r$ such that $Q\nsubseteq P^{c}_s$. Since $s<q$ we can conclude that $\widehat{q}(Q)\in U_A(Q)$. 

\item a-) Let $Q\subseteq P$ be an open set and $\widehat{q}(Q)\in\mathbb{Q}(Q)(Q)$ be such that $\widehat{q}(Q)\in L_A(Q)$. Since $Q\subseteq P^{c}_q$,  $\{Q\cap P^{c}_r\}_{r\in\mathbb{Q},r>q}$ is an open cover of $Q$; indeed

\[
\bigcup_{r\in\mathbb{Q},r>q}(Q\cap P^{c}_r)=Q\cap(\bigcup_{r\in\mathbb{Q},r>q}P^{c}_r)=Q\cap (\bigcap_{r\in\mathbb{Q},r>q}P_r)^{c}=Q\cap P^{c}_q=Q.
\]

We have then $Q\cap P^{c}_r\subset P^{c}_r$, which implies $\widehat{r}(Q\cap P^{c}_r)\in L_A(Q\cap P^{c}_r)$.\\
b-) If $\widehat{q}(Q)\in U_A(Q)$ by definition there exists $r<q$ such that $Q\nsubseteq P^{c}_r$. Therefore given  $s\in\mathbb{Q}$, such that $r<s<q$ we have $\widehat{s}(Q)\in U_A(Q)$.

\item  $Q\subseteq P$ be an open set and  $\widehat{q}(Q),\widehat{r}(Q)\in\mathbb{Q}(Q)(Q)$ such that $q<r$. Suppose  $\nVdash_{Q} \widehat{q}(Q)\in L_A$ then $Q\nsubseteq P^{c}_q$. Thus for $r$ exists $q$ such that $q<r$ and $Q\nsubseteq P^{c}_q$, then $\widehat{r}(Q)\in U_A(Q)$. On the other hand if we suppose $\nVdash_Q\widehat{r}(Q)\in U_A$, then for all $s<r$, $Q\subseteq P^{c}_s$. Particularly $Q\subseteq P^{c}_q$ then $\widehat{q}(Q)\in L_A(Q)$.

\item   Let $Q\subset P$ be an open set and $\widehat{q}(Q)\in\widehat{\mathbb{Q}}(Q)(Q)$. Suppose $\widehat{q}(Q)\in L_A(Q)$, then $Q\subseteq P^{c}_q$ and for all $r<q$, $Q\subset P^{c}_r$ thus $\widehat{q}(Q)\notin U_A(Q)$. In the same way if we suppose $\widehat{q}(Q)\in U_A(Q)$, there exists $r<q$ such that $Q\nsubseteq P^{c}_r$, since $P^{c}_q\subseteq P^{c}_r$ we have $Q\nsubseteq P^{c}_q$. Thus $\widehat{q}(Q)\notin L_A(Q)$.
\end{enumerate}
Thus  each spectral resolution of the identity defines a real number in $V(P)$. \\

On the other hand if 
\[U(P):Op(P)\rightarrow \bigcup_{Q\subseteq P} V(Q)\]
\[L(P):Op(P)\rightarrow \bigcup_{Q\subseteq P} V(Q)\]
define a real number in $V(P)$ we can find an spectral resolution of the identity in the next way. Given $r\in\mathbb{Q}$ consider the set 
\[P_r=\bigcup\{P\in Op(X): \forall q>r, \Vdash_P \widehat{q}(P)\in U\}.\] 

\begin{lema}
\[\bigcap_{r\in\mathbb{Q}} P_r=\emptyset\]
\end{lema}

\begin{proof}
From (1) in the definition of a Dedekind cut, for any $P\in Op(X)$, $P\neq \emptyset$,  there exists an open cover $\{Q_i\}_{i\in I}$ of $P$ and $\widehat{q_i}(Q_i)\in \widehat{\mathbb{Q}}(Q_i)(Q_i)$ such that $\widehat{q_i}(Q_i)\in L(Q_i)$. On the other hand from (5) in the definition of a Dedekind cut we have that for any $T\subseteq Q_i$, $\nVdash_T \widehat{q_i}(T)\in U$. Summarising we have an open cover $\{Q_i\}_i$ of $P$ and $\widehat{q_i}(Q_i)\in\widehat{Q}(Q_i)(Q_i)$ such that for every open set $T\subseteq Q_i$, $\nVdash \widehat{q_i}(T)\in U$. Therefore we have proved:
\begin{equation}\label{sf1}
\Vdash_P \exists \widehat{q}(P)\in\widehat{\mathbb{Q}}(P)(\widehat{q}(P)\notin U),
\end{equation}
for all $P\in Op(X)$. If 
\[R=\bigcap_{r\in\mathbb{Q}} P_r\neq \emptyset\]
we have 
\[\Vdash_R\widehat{q}(R)\in U\] 
for all $\widehat{q}(R)\in\widehat{\mathbb{Q}}(R)(R)$ but this contradicts \ref{sf1}. 
\end{proof}

\begin{lema}
\[\bigcup_{r\in\mathbb{Q}} P_r=X\]
\end{lema}

\begin{proof}
Let $\sigma\in X$ and $P$ an open neighbourhood of $\sigma$. From (1) in the definition of a Dedekind cut there exists an open set $Q_i$, and $\widehat{q_i}(Q_i)\in\widehat{\mathbb{Q}}(Q_i)(Q_i)$ such that $\sigma\in Q_i\subseteq P$ and $\Vdash_{Q_i}\widehat{q_i}(Q)\in U$. Then by property (2) of the definition of a Dedekind cut for all $q>q_i$ we have $\Vdash_{Q_i} \widehat{q}(Q_i)\in U$. Therefore $\sigma\in P_{q_i}$, since $\sigma$ was arbitrary we have $\bigcup_{r\in\mathbb{Q}} P_r=X$.
\end{proof}

\begin{lema}
\[\bigcap_{q<r}P_r=P_q\]
\end{lema}

\begin{proof}
If $\sigma\in P_q$, there exists a neighbourhood $P$ of $\sigma$ such that $\Vdash_P \widehat{s}(P)\in U$ for all $s>q$. Let $r>q$ arbitrary, then $\Vdash_P \widehat{s}(P)\in U$ for all $s>r$, then $\sigma\in P_r$ for all $r>q$ thus $\sigma\in\bigcap_{q<r}P_r$. On the other hand if $\sigma\in \bigcap_{q<r}P_r$, for all $r>q$ there exists a neighbourhood $Q_r$ of $\sigma$ such that $\Vdash_{Q_r}\widehat{s}(Q_r)\in U$ for all $s>r$. Let $Q=\bigcap_{r>q}Q_r$, then $\Vdash_Q\widehat{s}(Q)\in U$ for all $s>q$ and $\sigma\in Q$. Then $\sigma\in P_q$.
\end{proof}

Given a real number $r\in\mathbb{R}$, define 
\[P_r=\bigcap_{q\in\mathbb{Q}, r<q} P_q.\]
From the above results it is clear that the projections $\{P'_r\}_{r\in\mathbb{R}}$ associated to the family of clopen subsets $\{P_r\}_{r\in\mathbb{R}}$ form a spectral resolution of the identity. Since  there is a correspondence between self-adjoint operators in $\mathcal{U}$ and spectral families in $\mathcal{U}$, we have  a correspondence between the real numbers of the cumulative hierarchy of variable sets over $X=S_{\mathcal{U}}$ and the self-adjoint operators in  $\mathcal{U}$. This is the correspondence proved by Takeuti in the context of Boolean  valued models of set theory, the advantage here is that we get an explicit construction of the real numbers in function of the spectral representation of the respective self-adjoint operator. This construction will be probably fundamental  to understand the relation of quantum expressions and its classical counterparts.\\
There is an alternative proof of  Takeuti's correspondence in this context that even if it is less constructive than the proof presented above, it has a theoretical interest and it will probably be useful in more general contexts; for the sake of completeness we include this proof here. We start with a result which is proved in a more general context in \cite{maclane} ( theorem 2, chapter VI), we give here the proof for the hierarchies of variable sets.

\begin{teor}
Let $X$ be a topological space, and $\mathbb{V}^{X}$ the hierarchy of variable sets constructed over $X$. Let $P\in Op(X)$, then there is a correspondence between the objects forced as real numbers in $V(P)$ and the real valued continuous functions over $P$.
\end{teor}

\begin{proof}
Let $P\in Op(X)$,  $\lambda \in P$ and 
\[U(P):Op(P)\rightarrow \bigcup_{Q\subseteq P} V(Q)\]
\[L(P):Op(P)\rightarrow \bigcup_{Q\subseteq P} V(Q)\]
such that define a real number in $V(P)$. Consider 
\[L_{\lambda}=\{q\in \mathbb{Q}: \exists Q\subseteq P \text{ open set  s.t. }\lambda\in Q,\quad \widehat{q}(Q) \in L(P)(Q)\}\]
\[U_{\lambda}=\{r\in\mathbb{Q}: \exists Q\subseteq P\text{ open set s.t. }\lambda \in Q,\quad \widehat{r}(Q)\in U(P)(Q)\}\]
From the properties 1-5 above we have that the above sets defined a Dedekind cut in the classical sense  then a real number $a_{\lambda}=(L_{\lambda},U_{\lambda})$. Define the function 
\[f_{L,U}: P\rightarrow \mathbb{R}\]
such that $f_{L,U}(\lambda)=a_{\lambda}$. To show that the function $f_{L,U}$ is continuous consider an interval $(q,r)$ with $q,r\in\mathbb{Q}$ and $\lambda\in f^{-1}_{L,U}((q,r))$. We have in particular that $q\in L_{\lambda}$ and $r\in U_{\lambda}$ since $q<f_{L,U}(\lambda)=a_{\lambda}<r$, then there exists $Q$ and $T$ neighbourhoods of $\lambda$ such that $\widehat{q}(Q)\in L(P)(Q)$ and $\widehat{r}(T)\in U(P)(T)$. Thus, for every $\sigma\in Q\cap T$,  we have    $q\in L_{\sigma}$ and $r\in U_{\sigma}$ thus $\sigma\in f^{-1}((q,r))$  therefore $Q\cap T\subseteq f^{-1}_{L,U}((q,r))$. Since $\lambda$ was arbitrary and the rational intervals form a base we have proved the continuity of $f_{L,U}$.\\
Consider now $f:P\rightarrow \mathbb{R}$ a continuous function and
\begin{align*}
U_f(P):Op(P)& \rightarrow \bigcup_{Q\subseteq  P}V(Q)\\
Q & \mapsto \{\widehat{q}(Q)\in \widehat{\mathbb{Q}}(Q)(Q): \forall \lambda\in Q,  q>f(\lambda)  \}
\end{align*}
\begin{align*}
L_f(P):Op(P)& \rightarrow \bigcup_{Q\subseteq P}V(Q)\\
Q & \mapsto \{\widehat{q}(Q)\in \widehat{\mathbb{Q}}(Q)(Q): \forall \lambda \in Q,q< f(\lambda) \}
\end{align*}  
We want to prove that $L_f(P), U_f(P)$ define a real number in $V(P)$ i.e. that the conditions 1-5 above are satisfied. Let $Q\subseteq P$ an open set:\\
1. Let $Q_n=\{\lambda \in Q: -n<f(\lambda)<n\}$, then $\{Q_n\}_{n\in\mathbb{N}}$ form an open cover of $Q$, and we have $\widehat{-n}(Q_n)\in  L_f(P)(Q_n)$ and $\widehat{n}(Q_n)\in U_f(P)(Q_n)$.\\
2. Follows directly from the definitions of $L_f$ and $U_f$.\\
3.  Let $T\subseteq Q$ and $\widehat{q}(T)\in L_f(P)(T)$, then for any $\lambda\in T$ we have $q<f(\lambda)$. Since $f$ is continuous there exists a neighbourhood of $\lambda$,  $T_{\lambda}$ and $r_{\lambda}\in\mathbb{Q}$ such that for any $\sigma\in T_{\lambda}$ we have $q< r_{\lambda}<f(\sigma)$, then $\widehat{r}_{\lambda}\in L_f(P)(T_{\lambda})$. (Analogously for $U_f(P)$).\\
4. Let $T\subseteq Q$ and $q, r\in\mathbb{Q}$ such that $q<r$. Consider $T_1=\{\lambda\in T: q<f(\lambda)\} $ and $T_2=\{\lambda\in T: f(\lambda)<r\}$. We have $T=T_1\cup T_2$ and $\widehat{q}(T_1)\in L_f(P)(T_1)$, $\widehat{r}(T_2)\in U_f(P)(T_2)$.\\
5. Clear from the definition.
\end{proof}

\begin{corol}
Let $\mathcal{U}$ an abelian Von Neumann subalgebra of operators of a Hilbert space $H$, and $X=S_{\mathcal{U}}$ the Gelfand spectrum of $\mathcal{U}$ with the topology given by the sets $\{P\}_{P'\in P(\mathcal{U})}$ as defined above. There is a correspondence between the real numbers of the hierarchy $\mathbb{V}^{X}$ and the self adjoint operators of $\mathcal{U}$.
\end{corol} 

\begin{proof}
Let $A\in \mathcal{U}$ be a self adjoint operator and as before $\{P'_{r}\}_{r\in \mathbb{R}}$ the associated spectral family. Define the function $f_A: X\rightarrow \mathbb{R}$ such that $f_A(\lambda)=\lambda(A)$. We want to show that $f_A$ is continuous. Consider $q,r\in\mathbb{R}$ such that $q<r$, from the construction of the spectral family (see \cite{kadison} theorem 5.2.2) we know $f_A^{-1}((q,r))=P_r\setminus P_q$ which is an open set, then $f_A$ is a continuous function and by the above  theorem  it is associated with a real number in $V(P)$ for each $P\in Op(X)$.\\
Consider now $U, L$ such that 
\[U(P):Op(P)\rightarrow \bigcup_{Q\subseteq P} V(Q)\]
\[L(P):Op(P)\rightarrow \bigcup_{Q\subseteq P} V(Q)\]
define a real number for each $P\in Op(X)$, then from the above theorem we know there exists a continuous function $f:X\rightarrow \mathbb{R}$ associated to $L,U$. We know that the sets $\{P\}_{P'\in P(\mathcal{U})}$ are clopen sets also in the weak* topology over $X$ (see again \cite{kadison} theorem 5.2.2), then the function $f$ is also continuous in the weak* topology. Therefore by the Gelfand representation theorem (\cite{kadison} Theorem 4.4.3) to $f$ corresponds an operator $B_f$ such that $f(\lambda)=\lambda(B_f)$ for all $\lambda\in X$. Since $f(\lambda)=\lambda(B_f)\in \mathbb{R}$ for all $\lambda\in X$, we have that $B_f$ is a self adjoint operator (see \cite{kadison} theorem 4.3.8).
\end{proof}

The above results show how in this model deep results of operator theory recover a new interesting meaning. The value of this  new perspective is that it  probably contains the way to obtain a new picture of QM with a definite interpretation. Lets see why.\\

\indent Given $a\in\mathbb{R}$, $\widehat{a}(P)$ is a constant Dedekind cut over $P$, given by :
\[U_{\widehat{a}}(Q)=\{\widehat{q}(Q)\in \widehat{\mathbb{Q}}(Q)(Q): a<q\}\]
\[L_{\widehat{a}}(Q)=\{\widehat{q}(Q)\in \widehat{\mathbb{Q}}(Q)(Q): q<a\},\]
with $Q\subseteq P$ open. On the other hand, consider the spectral family $\{Q'_r\}_{r\in\mathbb{R}}$ such that $Q'_r=0$ if $r<a$ and $Q'_r=I$ if $r\geq a$, then the real number $(L,U)$ associated to this spectral family over an open set $P$ satisfies:

\begin{align*}
U(Q)&=\{\widehat{q}(Q)\in \widehat{\mathbb{Q}}(Q)(Q): \exists r\in\mathbb{Q}, r<q,  Q\nsubseteq P_r^{c} \}\\
&=\{\widehat{q}(Q)\in \widehat{\mathbb{Q}}(Q)(Q): a<q\}
\end{align*}

for any  open set $Q\subseteq P$ since   $Q\nsubseteq P^{c}_a=\emptyset$ for any $Q$ non empty. In the same way we have $L(Q)=\{\widehat{q}(Q)\in \widehat{\mathbb{Q}}(Q)(Q): q<a\}$ for any $Q\subseteq P$. Therefore $\widehat{a}(P)$ is the real number associated to the operator defined by $\{Q'_r\}_{r\in\mathbb{R}}$ (which is the operator $aI$).\\
As above consider again a self-adjoint operator $A\in\mathcal{U}$, $\{P'_r\}_{r\in\mathbb{R}}$ the spectral family associated to $A$ and  $(L_A,U_A)$ the real number defined by $A$. In $P_a$, with  $a\in\mathbb{R}$, consider $Q\subseteq P_a$. Since $Q\nsubseteq P^{c}_a$ we have 

\[
\{\widehat{q}(Q)\in \widehat{\mathbb{Q}}(Q)(Q): a<q\}\subseteq U_A(Q)=\{\widehat{q}(Q)\in \widehat{\mathbb{Q}}(Q)(Q): \exists r\in\mathbb{Q}, r<q,  Q\nsubseteq P_r^{c} \}.
\]

In the same way considering $P^{c}_a$ and $Q\subseteq P^{c}_a$, since $P^{c}_a\subseteq P^{c}_q$ for any $q<a$ we have  $Q\subseteq P^{c}_q$ for any $q<a$; then

\[\{\widehat{q}(Q)\in \widehat{\mathbb{Q}}(Q)(Q): q<a\}\subseteq L_A(Q)=\{\widehat{q}(Q)\in \widehat{\mathbb{Q}}(Q)(Q):   Q\subseteq P_q^{c} \}.
\]

Now remember that in the classical sense if we have two real numbers $a,b\in\mathbb{R}$ with $a=(L_a,U_a)$ and $b=(L_b, U_b)$ we define $a\leq b$ if and only if $U_b\subseteq U_a$ and $a\geq b$ if and only if $L_b\subseteq  L_a$. Applying this definition to our model the results above translate as
\[\Vdash_{P_a} A\leq \widehat{a}(P_a)\]
and
\[\Vdash_{P^{c}_a} \widehat{a}(P^{c}_a)\leq A\]
respectively, where we have denoted the real number defined by the self-adjoint operator $A$ using the same symbol. Remember from section \ref{logic} that $[[A\leq \widehat{a}]]_X$, $[[\widehat{a}\leq A]]_X$ represent the set of elements of $X$ where the respective propositions hold, since $P_a\cup P^{c}_a= X$ we have shown that
\[P_a=[[A\leq \widehat{a}]]_X\]
\[P^{c}_a=[[\widehat{a}\leq A]]_X.\]

Let $c,d\in\mathbb{R}$ such that $c<d$, we have

\begin{align*}
 [[\widehat{c}\leq A \leq \widehat{d}]]_X&=[[(\widehat{c}\leq A)\wedge(A\leq \widehat{d})]]_X\\
 &=[[\widehat{c}\leq A]]_X\cap [[A\leq \widehat{d}]]_X\\
 &=P^{c}_c\cap P_d\\
 &=P_d\setminus P_c.
 \end{align*}
 
The proof of this result was given  in the context of boolean valued models  in \cite{ozawa} (theorem 6.2), comparing it with the proof here presented, it becomes clear that the sheaf based approach gives an important simplification.\\

As before we denote by $[[\widehat{c}\leq A \leq \widehat{d}]]_X'$ the projector associated to the open set $[[\widehat{c}\leq A\leq \widehat{d}]]_X$. If $h\in H$ is a state vector,  we know from the classical formalism of QM that the probability of the  observable $A$ to assume some value in the interval $[c, d]$ in the state $h$ is given by 
 
 \[ || [[\widehat{c}\leq A\leq \widehat{d} ]]_X' h ||^{2}=||(P'_d-P'_c)h||^{2}, \]

but within our model we see that $h$ is just measuring the set $[[\widehat{c}\leq A\leq \widehat{d}]]_X$ of histories or universes where the proposition $\widehat{c}\leq A\leq \widehat{d}$ holds. Indeed, each quantum state $h\in H$ defines a measure $\mu_{h}$ over $X$ given by:
\begin{align*}
\mu_{h}:Op(X)&\rightarrow \mathbb{R}\\
P & \mapsto ||P'h||^{2}. 
\end{align*}
As we just saw these open sets are of the form $P=[[\varphi]]_X$, where $\varphi$ is a proposition about the quantum system, thus $P$ is the set of histories where the proposition $\varphi$ is verified and $\mu_h$ is measuring the proportion of histories where any proposition is verified. This is literally the interpretation given to quantum states in the Deutsch-Everett multiverse interpretation of Quantum Mechanics. It is important to note that the type of propositions about the Quantum System are not limited just  to propositions referring to the values of physical variables as in the example above; in this case the results are valid for any kind of propositions about the quantum system that can be expressed with the language of set theory, probably even propositions about emergent properties related to the conception of classical spacetime.\\

\subsection{Generic Models and the Emergence of Classical Reality}

Any new mathematical model that will  improve our understanding of QM has to explain or give us a hint of how the classical reality of our everyday experience emerges from the deep Quantum reality of the elementary particles that constitute everything. Surprisingly, in the models developed above there is a way to collapse the multiversal model of Quantum Variable sets to a classical model in such a way that what is perceived by the classical universe is conditioned by the structure of the multiversal model.  This can be done over any Sheaf of structures, however, for the sake of clarity,  we will explain how this process work on the hierarchy of Variable Sets. The main idea in this section will be the notion of generic model and genericity. This notion was originally introduced  by Paul Cohen in his works on the independence of the Continuum Hypothesis and the Axiom of Choice; becoming then fundamental in modern set theory and model theory. In the context of sheaves of structures it is possible to introduce the notion of generic model and genericity in a more general framework which unifies previous approaches, and where hypothesis about enumerability are not required (see \cite{caicedo}). We will see then how to connect these ideas related to the foundations of mathematics to Quantum Mechanics via the hierarchy of Quantum Variable Sets.\\

We start with an easy definition.

\begin{defin} Let $X$ be a topological space and  $Op(X)$ the set of open sets of $X$. Let  $\mathcal{F}\subset
Op(X)$, $\mathcal{F}$ is called a filter of open sets  of  $X$ if:\\
i. $X\in\mathcal{F}$.\\
ii. If  $U,V\in\mathcal{F}$ then  $U\cap V\in\mathcal{F}$.\\
iii.Given  $V\in\mathcal{F}$ If $V\subset U$ then
$U\in\mathcal{F}$.
\end{defin}

Filters over the base space $X$ of a Sheaf of Structures will be the tools we will use to collapse the intuitionistic sheaf to a classical model. To do this we need an special kind of filters where the essence of being classical is captured in such a way that for each proposition the excluded middle holds in an open set of the filter, and for each  existential proposition there exists an open set where the existential is verified in the classical sense. In other words:

\begin{defin} 
Let  $\mathcal{F}$ be a  filter of open sets of $X$, we say that $\mathcal{F}$ is a Generic Filter of 
$\mathbb{V}^{X}$ if:\\
i. Given $\varphi(v_1,...,v_n)$ a first order formula (a proposition) and  $\sigma_1,...,\sigma_n$ arbitrary sections of $\mathbb{V}^{X}$ defined on $P\in\mathcal{F}$, there exists 
$Q\in\mathcal{F}$ such that 
\[\Vdash_Q\varphi[\sigma_1,...,\sigma_n] \text{ or }
 \Vdash_Q\neg\varphi[\sigma_1,...,\sigma_n].\]
ii. Given $\sigma_1,...,\sigma_n$ arbitrary sections of\quad  
$\mathbb{V}^{X}$ defined  on  $P\in\mathcal{F}$, and 
$\varphi(v,v_1,...,v_n)$ a first order formula. If
$\Vdash_P\exists v\varphi(v,\sigma_1,...,\sigma_n)$
then there exist $Q\in\mathcal{F}$ and $\sigma$ defined on $Q$ such that
$\Vdash_Q\varphi(\sigma,\sigma_1,...,\sigma_n)$.
\end{defin}

The next important result shows that generic filters exist, the proof of this result is contained in \cite{caicedo}.

\begin{teor}\label{generic}
A filter is a generic filter of $\mathbb{V}^{X}$  if and only if  it is a maximal filter of $X$.
\end{teor}

Using a generic filter we can construct a classical model from a sheaf of structures.

\begin{defin}  
We can associate a classical model to $\mathbb{V}^{X}$ in the next way:\\
let
\[\mathbb{V}^{X}[\mathcal{F}]=\lim_{\rightarrow
P\in\mathcal{F}}V(U)\] i.e
\[\mathbb{V}^{X}[\mathcal{F}]=\dot{\bigcup}_{P\in\mathcal{F}}V(P)/_{\sim_{\mathcal{F}}}\]
 where for  $\sigma\in V(P)$ and $\mu\in V(Q)$
 \[\sigma\sim_{\mathcal{F}}\mu\Leftrightarrow\text{ there exists
 }R\in\mathcal{F}\text{ such that 
 }\sigma\upharpoonright_R=\mu\upharpoonright_R.\]
 Let $[\sigma]$ be the class of  $\sigma$. We define  relations and functions in the next way:
 \[([\sigma_1],...,[\sigma_n])\in \mathcal{R}^{\mathbb{V}^{X}[\mathcal{F}]}
 \Leftrightarrow\text{  there exists
 }U\in\mathcal{F}:(\sigma_1,...,\sigma_n)\in R^{V(U)}\]
 \[f^{\mathbb{V}^{X}[\mathcal{F}]}([\sigma_1],...,[\sigma_n])=
 [f^{V(U)}(\sigma_1,...,\sigma_n)]\] $\mathcal{F}.$ If $\mathcal{F}$ is a generic filter over $X$ for  $\mathbb{V}^{X}$ we say that 
 $\mathbb{V}^{X}[\mathcal{F}]$ is a generic model. 
 \end{defin}

 What is perceived by these classical universal models depends on the structure of the multiversal sheaf of variable sets. This is clearly expressed in the following result (see \cite{caicedo}).

\begin{teor}
Let $\mathcal{F}$ be a generic filter over $X$ for 
$\mathbb{V}^{X}$, then:
\begin{eqnarray*}
\mathbb{V}^{X}[\mathcal{F}]\models\varphi([\sigma_1],...,[\sigma_n])&\Leftrightarrow&\text{
there exists }P\in\mathcal{F}\text{ such that
} \Vdash_P\varphi^G(\sigma_1,...,\sigma_n)\\
&\Leftrightarrow&\{\lambda\in
X:\Vdash_{\lambda}\varphi^G(\sigma_1,...,\sigma_n)\}\in\mathcal{F},
\end{eqnarray*}
where $\varphi^G$ is the G\"odel translation\footnote{See \cite{benavides} section 0.2.} of the formula $\varphi$.
\end{teor}

It is not important here to know what is precisely the G\"{o}del translation, it is enough to know that $\varphi^{G}$  is a form to reformulate $\varphi$ in a classical equivalent way, for example, the G\"{o}del translation  of the proposition $(a=b)$ is $(a=b)^{G}=\neg\neg (a=b)$ which in  the sense of classical logic is equivalent. The important fact to know  is that a proposition can be deduced classically if and only if its G\"{o}del translation can be deduced intuitionistically. 

\begin{corol} Let $X$ be a topological space and $\mathcal{F}$ a generic filter of $\mathbb{V}^{X}$, then 
\[\mathbb{V}^{X}[\mathcal{F}]\models ZF.\]
\end{corol} 

As an example  we can choose the base space $X$ in such a way that in the collapsed model of $\mathbb{V}^{X}$ the Axiom of Choice is valid but the Continuum  hypothesis does not hold, this is the intrinsic essence of  Paul Cohen's proof (see \cite{benavides}, \cite{benavides3} \cite{cohen}).\\

To see how these results can be applied to Quantum Mechanics, lets consider again the cumulative hierarchy of quantum variable sets over the spectrum $X=S_\mathcal{U}$ of an abelian Von Neumann sub algebra of the algebra $B(H)$ of bounded operators of a Hilbert space $H$. The  important result  here is that each history $\lambda \in X$ determines a Generic Filter.

\begin{teor} Considering $\lambda \in X=S_{\mathcal{U}}$, the set 
\[\mathcal{F}_{\lambda}=\{P\in Op(X): \lambda(P')=1\},\]
where $P'$ is the projection operator  associated to the open set $P$, is a generic filter of $\mathbb{V}^{X}$.
\end{teor} 

\begin{proof}
Remember that  $P=\{\sigma\in X: \sigma(P')=1\}$. If $P, Q\in \mathcal{F}_{\lambda}$, we have  $\lambda(P')=1=\lambda(Q')$, then 
\[\lambda((P\cap Q)')=\lambda(P'\wedge Q')=\lambda(P'Q')=\lambda(P')\lambda(Q')=1.\] Then $P\cap Q\in \mathcal{F}_{\lambda}$. On the other hand if $R\in \mathcal{F}_{\lambda}$ and $S\in Op(X)$ is such that $R\subset S$, this implies that $\lambda \in S$, therefore $\lambda(S')=1$ then $S\in \mathcal{F}_\lambda$. Finally as $X'=I$ where $I$ is the identity operator, we have  for any operator $P'$ and any $\lambda$ in $X$
\[\lambda(P')=\lambda(IP')=\lambda(I)\lambda(P'),\]
then $\lambda(I)=1$, thus $X\in F_{\lambda}$. We have  shown that $\mathcal{F}_{\lambda}$ is a filter, to see that it is maximal suppose that $\mathcal{G}$ is a filter such that $\mathcal{F}_{\lambda}\subset \mathcal{G}$. Let $P\in \mathcal{G}\setminus \mathcal{F}_{\lambda}$, then $\lambda (P')=0$. Therefore 
\[\lambda((P^{c})')=\lambda((X\setminus P)')=\lambda(I)-\lambda(P')=1,\]
then $P^{c}\in \mathcal{F}_{\lambda}$. So we conclude that $P^{c}\in\mathcal{G}$ and $\emptyset=P\cap P^{c}\in \mathcal{G}$, then $\mathcal{G}=Op(X)$. Thus $F_{\lambda}$ is a maximal filter and by theorem \ref{generic}  a generic filter.
\end{proof}

This last result tells us that to each history corresponds a classical universe, the collapsed universe $\mathbb{V}^{X}[\mathcal{F}_{\lambda}]$ obtained from the sheaf $\mathbb{V}^{X}$ via the generic filters $\mathcal{F}_{\lambda}$. In  collapsed universes the propositions about the Quantum System assume  definite truth values as in classical reality. And these truth values depend on the structure of similar histories, which is represented by the filters $\mathcal{F}_{\lambda}$ which structures depend themselves on the structure of the Hilbert Space from which they arise.  To understand how this process works we will need to complete the reformulation of Quantum Mechanics within these models, but if this approach to understand emergence of classical reality  is correct, it will reflect a new way to understand  emergence in physics. In this case we will not have the emergence as a weak classical limit where some parameters as velocity, scale or gravity tend to certain values but we will have a kind of emergence more close to the sense of emergence as intended in computation or biology, where a lower level structure, represented in this case by a sheaf of structures, determines what is perceived by a higher level structure represented in this case by the collapsed Generic Models.\\

 Any satisfactory reformulation of QM should have three fundamental characteristics; first it has to settle a definite interpretation of the theory, second it has to  explain the emergence of classical reality and finally it has to be flexible enough to be extended to General Relativity. In this work we saw that the model  presented  here is a good candidate to fulfil the first two conditions. And probably these kind of models are  perfect to include GR from a totally new perspective (see \cite{bell} for example). All these results show that probably we  really need a formalism founded over a quantum logic to understand quantum reality and to construct quantum gravity theories.

\end{document}